\documentclass[11pt]{article}

\usepackage{amsmath,amscd,amsfonts,eucal,latexsym,amssymb,mathrsfs}
\usepackage{amsxtra, amsthm, calc}
\theoremstyle{definition}
\newtheorem{Definition}{Definition}[section]
\theoremstyle{plain}
\newtheorem{Theorem}[Definition]{Theorem}
\newtheorem{Proposition}[Definition]{Proposition}
\newtheorem{Lemma}[Definition]{Lemma}
\newtheorem{Corollary}[Definition]{Corollary}
\theoremstyle{remark}
\newtheorem{Remark}[Definition]{Remark}
\newtheorem*{Remarks}{Remarks}
\numberwithin{equation}{section}
\newcounter{remcount}

\newenvironment{remlist}{\begin{list}{(\roman{remcount})}{\usecounter{remcount}\setlength{\leftmargin}{0 cm}\setlength{\rightmargin}{0 cm}\setlength{\topsep}{0 cm}\setlength{\itemsep}{0 cm}\setlength{\parsep}{\parskip}\setlength{\labelwidth}{1 cm}\setlength{\labelsep}{0.5 em}\setlength{\itemindent}{1 cm + 0.5 em}}}{\end{list}}
\def\cA{{\cal A}}
\def\cB{{\cal B}}

\def\cE{{\cal E}}
\def\cF{{\cal F}}

\def\cH{{\cal H}}

\def\bC{{\mathbb C}}

\def\bN{{\mathbb N}}

\def\bR{{\mathbb R}}

\def\a{\alpha}
\def\b{\beta}
        
\def\d{\delta}        \def\D{\Delta}
  
\def\k{\kappa}
\def\l{\lambda}       

\def\n{\nu}

\def\r{\rho}
\def\s{\sigma}

\def\t{\tau}

\def\o{\omega}        \def\O{\Omega}

\def\supp{{\text{supp}\,}}

\DeclareMathOperator{\aut}{Aut}
\newcommand{\Mor}{\mathrm{Mor}}
\newcommand{\Id}{{\bf 1}}
\newcommand{\slim}{\operatorname{s-lim}\displaylimits}
\newcommand{\sslim}{\operatorname{s^*-lim}\displaylimits}
\newcommand{\wlim}{\operatorname{w-lim}\displaylimits}
\newcommand{\uA}{\underline{A}}
\newcommand{\uB}{\underline{B}}
\newcommand{\uC}{\underline{C}}
\newcommand{\uU}{{\underline{U}}}
\newcommand{\uV}{{\underline{V}}}
\newcommand{\uW}{{\underline{W}}}

\newcommand{\uFF}{\underline{\mathfrak F}}
\newcommand{\uAA}{\underline{\mathfrak A}}
\newcommand{\uBB}{\underline{\mathfrak B}}
\newcommand{\ua}{\underline{\alpha}}

\newcommand{\ur}{\underline{\rho}}

\newcommand{\us}{\underline{s}}
\newcommand{\uooi}{\underline{\o}_{0,\iota}}

\newcommand{\Foi}{\mathcal{F}_{0,\iota}}
\newcommand{\Aoi}{\mathcal{A}_{0,\iota}}
\newcommand{\uo}{\underline{\omega}}

\newcommand{\poi}{\pi_{0,\iota}}
\newcommand{\Hoi}{{\cal H}_{0,\iota}}
\newcommand{\Ooi}{\Omega_{0,\iota}}
\newcommand{\ooi}{\omega_{0,\iota}}
\newcommand{\uAAac}{\uAA^\bullet}
\newcommand{\uFFac}{\uFF^\bullet}
\newcommand{\poiac}{\poi^\bullet}
\newcommand{\rac}{\boldsymbol{\rho}^\bullet}
\newcommand{\pac}{\boldsymbol{\phi}^\bullet}
\newcommand{\Aoir}{\cA_{0,\iota,r}}
\newcommand{\br}{\boldsymbol{\rho}}
\newcommand{\bs}{\boldsymbol{\sigma}}
\newcommand{\bp}{\boldsymbol{\phi}}
\newcommand{\bpsi}{\boldsymbol{\psi}}
\newcommand{\AMor}{\mathrm{AMor}}
\newcommand{\AIso}{\mathrm{AIso}}
\newcommand{\Morps}{\Mor_{\mathrm{ps}}}
\renewcommand{\hom}{\mathrm{hom}}

\newcounter{propcount}

\newlength{\maxlabelwidth}
\newenvironment{proplist}[2][1]{\begin{list}{(\roman{propcount})}{\usecounter{propcount}\setcounter{propcount}{#2}\settowidth{\maxlabelwidth}{\textit{(\roman{propcount})}}\setcounter{propcount}{#1-1}\setlength{\leftmargin}{\maxlabelwidth+ 0.5 em}\setlength{\rightmargin}{0 cm}\setlength{\topsep}{0 cm}\setlength{\itemsep}{0 cm}\setlength{\parsep}{\parskip}\setlength{\labelwidth}{\maxlabelwidth}\setlength{\labelsep}{0.5 em}\setlength{\itemindent}{0 cm}}}{\end{list}}

\newcommand{\inst}[1]{$^\textrm{#1}$ }

\begin{document}

\title{Asymptotic morphisms and superselection theory\\ in the scaling limit}
\author{Roberto Conti\inst{1} \and Gerardo Morsella\inst{2}}
\date{
\parbox[t]{0.9\textwidth}{\footnotesize{%
\begin{itemize}
\item[1] Dipartimento di Scienze di Base e Applicate per l'Ingegneria, sezione di Ma\-te\-ma\-ti\-ca, Sapienza Universit\`a di Roma, via A. Scarpa, 16 I-00161 Roma (Italy), e-mail: roberto.conti@sbai.uniroma1.it
\item[2] Dipartimento di Matematica, Universit\`a di Roma Tor Vergata, v.le della Ricerca Scientifica, 1 I-00133 Roma (Italy), e-mail: morsella@mat.uniroma2.it
\end{itemize}
}}
\\
\vspace{\baselineskip}
\today}

\maketitle

\begin{abstract}
Given a local Haag-Kastler net of von Neumann algebras and one of its scaling limit states, we introduce a variant of the notion of asymptotic morphism by Connes and Higson, and we show that the unitary equivalence classes of (localized) morphisms of the scaling limit theory of the original net are in bijection with classes of suitable pairs of such asymptotic morphisms. In the process, we also show that the quasi-local C*-algebras of two nets are isomorphic under very general hypotheses, and we construct an extension of the scaling algebra whose representation on the scaling limit Hilbert space contains the local von Neumann algebras. We also study the relation between our asymptotic morphisms and superselection sectors preserved in the scaling limit.
\end{abstract}

\section{Introduction}
Quantum Field Theory (QFT) provides a convenient mathematical framework for the description of fundamental interactions, and it is widely recognized that, in its continual search of concepts and results, it owes a lot to geometrical insight.

A fruitful approach to QFT, that we follow in this work, is the algebraic one~\cite{Ha}, in which a theory is determined by an assignement $O \mapsto \cA(O)$,
also called \emph{net},
between spacetime regions and operator algebras,
satisfying a number of physically motivated axioms. On the other hand, a very powerful mathematical arena for developing geometry in the language of operator algebras is provided by Noncommutative Geometry (NCG)~\cite{Co}. 

Possible connections between QFT and NCG have been explored by several authors in the recent past, among which we just cite~\cite{JLO, DFR, Lo, BG, CHKL, CCHW}. In the present paper, we look at some NCG inspired structures naturally arising in the algebraic framework of QFT. In particular, the bridge between QFT and NCG is provided here by (a suitable variant of) the notion of asymptotic morphisms originally proposed by Connes and Higson~\cite{CH} as a starting point for a version of KK-theory, thereby called E-theory, possibly more suitable for applications to the Baum-Connes conjecture.
However, the need for such a notion does not arise from the interplay between ``classical'' and ``quantum'' but rather from a comparison between the charge content of a theory - described by equivalence classes of representations of $\cA$, called \emph{superselection sectors} - at different scales. 
(In other words, our small parameter is the spatio-temporal scale $\lambda$ and not Planck's constant $\hbar$ as in~\cite{MT}.)
On the physical side, this is related to the idea that charges can be confined at small scales, which is believed to happen, e.\ g., for the color charge in Quantum Chromodynamics (QCD).


\medskip

In order to explain in greater detail the way in which (objects similar to) asymptotic morphisms appear in QFT, we need to recall two independent constructions that are relevant for the present work. On one hand
(if the dimension of spacetime is at least three),
by the results of Doplicher and Roberts~\cite{DR} one can construct
in a canonical way a field net $O \mapsto \cF(O)$ and a compact gauge group $G$
acting on $\cF$, such that $\cA$ is recovered as the  fixed point of $\cF$ under $G$, and the Hilbert space on which $\cF$ naturally acts contains all superselection sectors of $\cA$.
On the other hand,  it is well-known that the understanding of some important issues in physics, such as color confinement in QCD, requires a careful study of the short distance behaviour of quantum fields~\cite{Buc1}. With this aim in mind, Buchholz and Verch~\cite{BV} suggested how to associate to $\cA$ a whole family of new nets $\Aoi$, hereafter called \emph{scaling limit nets}, which are in a sense ``tangent spaces''
of $\cA$ and contain the whole information of the theory at short scales, in an intrinsic way.

In general, the two constructions do not commute, namely the scaling limit of the field net of $\cA$ is not necessarily equal to the field net of the scaling limit $\Aoi$.
To make a long story short, sectors of $\cA$ and $\Aoi$ are related in some nontrivial way~\cite{DMV, DMV2, DM}: in general, sectors of $\cA$ may either disappear in the scaling limit, or be preserved, while new sectors of $\Aoi$ may arise in the limit. The latter ones are identified as describing the confined charges of the underlying theory $\cA$ for obvious reasons.
It is also interesting to observe that the mutual relationship between the two canonical gauge groups
$G$ and $G_{0,\iota}$ of $\cA$ and $\Aoi$ could possibly play a role in the process of obtaining an intrinsic description of quantum gauge theories.

Thanks to Haag duality, the superselection sectors of a local net are described by
(classes of) suitable endomorphisms of the associated quasi-local algebra~\cite{DHR}.
This clearly applies both to $\cA$ and $\Aoi$. Now the problem naturally arises about finding ways to conveniently describe all the sectors of $\Aoi$ in terms of structural aspects of the original net $\cA$ at scale one. It has been known for some time that all sectors of $\Aoi$ can be described
by the so-called \emph{asymptotic charge transfer chains} of $\cA$~\cite{M1,M2}, a concept arising by applying the usual method of ``shifting a compensating charge behind the moon'' to the scaling limit net.
Among these sectors, those arising from the preserved sectors of $\cA$ have been described in somewhat more detail. In particular, these sectors of $\Aoi$ can still be represented through families $(\r(\l))_{\l > 0}$ of endomorphisms of $\cA$, parametrized by the scale $\l$. This observation led S.~Doplicher to ask to one of us the question whether it is possible to describe all the superselection sectors of $\Aoi$ by objects related to the Connes-Higson asymptotic morphisms.

The main motivation for this paper has been to provide a first answer this question. It turns out that this can actually be done, but, due to the peculiarities of the scaling algebra construction, even if the asymptotic morphisms we consider are formally similar to the Connes-Higson ones, strictly speaking they are not exactly the same thing, see Definition~\ref{def:tameasympmor}.

As explained in more detail at the beginning of Section~\ref{sec:dilation}, asymptotic morphisms of $\cA$ can be equivalently defined as morphisms
\begin{equation}\label{eq:br}
\br: \cA \to \Aoi
\end{equation}
compatible in some mild way with the net structure and Heisenberg uncertainty principle.
\emph{Bona fide} asymptotic morphisms $(\r_\l)_{\l > 0}$
of $\cA$ are thus obtained after choosing a continuous section for a suitable extension of the GNS representation $\poi$ of a scaling limit state. The construction of the latter is carried out in Section~\ref{sec:liftrep}. It provides us with an extension of the scaling algebra whose representation on the scaling limit Hilbert space contains the local von Neumann algebras $\Aoi(O)$.

Then, in order to set up a correspondence between asymptotic morphisms of $\cA$ and morphisms $\r_0$ of $\Aoi$, one needs to relate the latter ones to morphisms as in~\eqref{eq:br}. This can of course be done by defining $\br = \r_0\bp$ for an isomorphism $\bp: \cA \to \Aoi$ of the quasi-local C*-algebras. In some specific cases, the choice of such an isomorphism is more or less natural. This is certainly the case for dilation covariant theories satisfying some nuclearity requirement, discussed, as a warm up, in Section~\ref{sec:dilation}, and also for the bosonic
massive scalar free field in $d \geq 3$ spacetime dimensions, thanks to the locally Fock property.
More generally, it was already pointed out in ~\cite{Do} that a C*-algebra isomorphism $\bp: \cA \to \cB$ as above exists for a large class of nets $\cA$ and $\cB$ (we also discuss this issue in detail in Section~\ref{sec:iso}). 
However, there seems to be no natural way of singling out a particular such isomorphism, even in the particular case $\cB = \Aoi$.
Indeeed, loosely speaking, any such isomorphism could be considered in a sense analogous to the choice of renormalization scheme in conventional QFT. Therefore, in the spirit of the scaling algebra construction, we take the point of view of considering them all together.
As shown in Section~\ref{sec:asympmor}, this results  in a one to one correspondence between unitary equivalence classes of  morphisms of $\Aoi$ and suitably defined equivalence classes of \emph{pairs} of asymptotic morphisms of $\cA$  (thus in a way formally similar to the Cuntz description of KK-groups).

We note that, to the best of our knowledge, this is the first time that the existence of isomorphisms between the quasi-local C*-algebras of possibly different nets is promoted from a mere observation to a structural tool.

In the remaining part of the paper we discuss further properties of our asymptotic morphisms in relation with the superselection structures of $\cA$ and $\Aoi$. In particular, in Section~\ref{sec:localized} we define \emph{localized} pairs of asymptotic morphisms and characterize them in terms of asymptotic charge transfer chains. In Section~\ref{sec:preserved} we come back to our original motivation by showing that there is a natural choice of an asymptotic morphism corresponding to a preserved sector of $\cA$, obtained in terms of the family $(\r(\l))_{\l>0}$ of localized morphisms of $\cA$ mentioned above. These results raise the further question about the existence of some close connection between asymptotic morphisms and asymptotic charge transfer chains also for non preserved sectors. Although we refrain from giving a complete answer here, in Section~\ref{sec:ACTCremarks} we point out some difficulties which one has to face in trying to do so. Finally in Section~\ref{sec:category} we define a C*-category of asymptotic morphisms, following a similar construction for Connes-Higson asymptotic morphisms presented in Appendix~\ref{app:category}. The paper is ended by a summary of our results and an outlook about open issues and future directions.

\section{Asymptotic morphisms for dilation covariant theories}\label{sec:dilation}
According to E-theory, given two C$^*$-algebras $A$ and $B$, an asymptotic morphism from $A$ to $B$
determines a  $*$-homomorphism from $A$ into $B_0 := C_b(\bR_+,B)/C_0(\bR_+, B)$
and, conversely, all asymptotic morphisms arise in this way, making use of set-theoretic
sections from $B_0$ to $C_b(\bR_+,B)$~\cite{CH, GHT} (see also Appendix A).
In particular, asymptotic morphisms of a C$^*$-algebra $A$ into itself are basically
the same as morphisms $A \to C_b(\bR_+,A)/C_0(\bR_+,A)$.

In the present investigation, we adapt this notion to a special class of C$^*$-algebras endowed with
further structure motivated by Quantum Field Theory~\cite{Ha}. More precisely, we assume we are given a local net $O \mapsto \cA(O)$ of von Neumann algebras on (open double cones in) $d$-dimensional Minkowski space acting irreducibly on a vacuum Hilbert space $\cH$. We assume furthermore that $\cA$ is covariant with respect to an automorphic action of the translation group $\a : \bR^d \to \operatorname{Aut}(\cA)$
which is unitarily implemented and satisfies the spectrum condition. Finally, we denote by  $\O \in \cH$ the vacuum vector, i.e.\ the (up to a phase) unique translation invariant unit vector in $\cH$, and by $\o = \langle\O,(\cdot)\O\rangle$ the corresponding vector state on the quasilocal C$^*$-algebra of the net, still denoted by $\cA$.

In this framework,
a reasonable way to provide examples of suitable asymptotic morphisms of $\cA$ relies on the natural analogy between the algebra $A_0$ of the above remark and the scaling limit net $\Aoi$~\cite{BV} of $\cA$, and therefore on the existence of morphisms from $\cA$ into $\Aoi$. The role of $C_b(\bR_+,A)$ is then played by the quasi-local scaling algebra $\uAA$ of the scaling algebras net
\begin{equation}
\uAA(O) := \Big\{\uA \in \uBB \ | \ \uA_\l \in \cA(\l O), \l > 0, \lim_{x \to 0} \ua_x(\uA) = \uA \Big\} \ ,
\end{equation}
where $\uBB$ is the C$^*$-algebra of all bounded functions $\uB: \bR_+ \to B(\cH)$ with naturally defined operations and $\ua_x(\uA)_\l := \a_{\l x}(\uA_\l)$,
while the role of the ideal $C_0(\bR_+,A)$ is played by the kernel of the GNS representation $\poi$
of a fixed scaling state $\uooi$ of $\uAA$. The latter is defined as a weak* limit point, as $\l \to 0$, of
the family $(\uo_\l)_{\l>0}$ of states over $\uAA$ defined by
$$\uo_\l(\uA) := \o(\uA_\l), \qquad \l > 0,\, \uA \in \uAA.$$
In view of applications to superselection theory, one should also stress that $\Aoi$, which is going to replace $A_0$, is better defined not as the quotient $\uAA/\ker\poi$ as in~\cite{BV}, but rather as
a local completion of the latter, meaning that
$$\Aoi(O) := \poi(\uAA(O))'' \qquad \text{for all }O.$$

Before proceeding further along these lines, we anticipate some considerations.
As discussed in the Introduction, the main motivation of the present paper is to look for a description of
superselection sectors of the scaling limit theory in terms of asymptotic morphisms
of the underlying theory. In general, however, it is not obvious at all how to associate
in a natural and physically meaningful way morphisms from $\cA$ to $\Aoi$ to
sectors of $\Aoi$ (or of $\cA$). Moreover a number of technical problems
are likely to arise if one is going to take seriously the analogy between $A_0$ and $\Aoi$, e.g., the fact that it is not clear if it is possible to extend
sections $\poi(\uAA) \to \uAA$ to $\Aoi$. 
Furthermore, note that in E-theory
the asymptotic notions are required to hold in the norm topology of $A$, while
in our context it is reasonable to expect that this condition needs to be weakened. This is
of course related to the fact that scaling algebra elements are not
norm continuous functions of $\l$ in general.

\medskip
In order to illustrate the connection between morphisms of the scaling limit net
and asymptotic morphisms of the underlying net, we consider a simple example in which some of such problems do not occur. Let $\cA$ be a net of C$^*$-algebras for which there
exists an isomorphism of the quasi-local C$^*$-algebras \begin{equation*}
\bp : \cA \to \poi(\uAA).
\end{equation*}
This is the case, for instance, of a dilation covariant theory satisfying Haag-Swieca compactness, where $\bp$ can be chosen to be even a net isomorphism~\cite{BV}. We pick a set-theoretic section $\us : \poi(\uAA) \to \uAA$ and a morphism $\rho_0 : \poi(\uAA) \to \poi(\uAA)$ and define
\begin{equation*}
\rho_\l(A) := \us(\rho_0\bp(A))_\l, \qquad A \in \cA,\,\l>0.
\end{equation*}

\begin{Proposition}
With $(\l_\k)_\k$ a net defining the scaling limit state $\uooi$, there holds, for all $A,B \in \cA$, $\a \in \bC$:
\begin{gather*}
\ur(A) := (\l \mapsto \rho_\l(A)) \in \uAA,\\
\lim_{\k} \| [\rho_{\l_\k}(A)^*-\rho_{\l_\k}(A^*)]\O\|= 0, \\
\lim_\k \| [\rho_{\l_\k}(A+\a B) - \rho_{\l_\k}(A) - \a \rho_{\l_\k}(B)]\O\|=0, \\
\lim_\k \| [\rho_{\l_\k}(AB) - \rho_{\l_\k}(A)\rho_{\l_\k}(B)]\O\|=0.
\end{gather*}
In particular setting $\br(A) = \poi(\ur(A))$, $A \in \cA$, one gets a morphism of C$^*$-algebras $\br : \cA \to \poi(\uAA)$.
\end{Proposition}

\begin{proof}
The first property is obvious from the definition of $\r_\l$. We just prove the fourth equation, the proof of the other two being analogous. Noting that, by definition, $\poi(\us(\bp(A))) = \bp(A)$, we have
\begin{equation*}\begin{split}
\lim_\k \| [\rho_{\l_\k}(AB) &- \rho_{\l_\k}(A)\rho_{\l_\k}(B)]\O\|\\
&=\lim_\k \|[\us(\rho_0\bp(AB))_{\l_\k}-\us(\rho_0\bp(A))_{\l_\k}\us(\rho_0\bp(B))_{\l_\k}]\O\| \\
&= \|[\poi(\us(\rho_0\bp(AB)))-\poi(\us(\rho_0\bp(A))\us(\rho_0\bp(B)))]\Ooi\| \\
&=\|[\rho_0\bp(AB)-\rho_0\bp(A)\rho_0\bp(B)]\Ooi\| = 0.
\end{split}
\end{equation*} The last statement follows at once from the equation $\br= \rho_0\bp$.
\end{proof}

In the case of a dilation covariant theory there is a canonical choice of the section $\us : \poi(\uAA) \to \uAA$, namely
\begin{equation*}
\us(\poi(\uA))_\l = \d_\l\bp^{-1}(\poi(\uA)).
\end{equation*}
Then given a morphism  $\r_0$ of $\poi(\uAA)$ and denoting by $\r = \bp^{-1}\rho_0\bp$ the associated morphism of $\cA$, the $(\r_\l)_{\l>0}$ corresponding to $\r_0$ satisfies
\begin{equation} \label{eq:rholdilation}
\rho_\l(A) = \d_\l \r(A) = \rho(\l)(\d_\l(A)), \qquad A \in \cA,
\end{equation}
where $\rho(\l) = \d_\l \r \d_\l^{-1}$ is the dilated morphism $\r$, which is localized in $\l O$ if
$\r$ is localized in $O$. 
This suggests the possibility of a similar relation in more general theories between asymptotic morphisms and families of morphisms $(\rho(\l))_{\l>0}$
associated to preserved sectors as in~\cite{DMV} (in a dilation covariant theory all sectors are preserved). We will show later that this is indeed the case, see Theorem~\ref{thm:preserved} and Remark~\ref{rem:preserved}\ref{it:dilation}. 

\section{On the existence of C$^*$-isomorphisms between quasi-local algebras}\label{sec:iso}
In this section we present a somewhat more accurate description of the C$^*$-isomorphism between the quasi-local C$^*$-algebras of possibly different local nets, whose existence can be traced back to~\cite{Do}. Such C$^*$-isomorphisms will be employed later in order to relate morphisms of the scaling limit net $\Aoi$ with morphisms between $\cA$ and $\Aoi$. In this way, we can circumvent one of the problems mentioned in Section~\ref{sec:dilation}.

We recall that a local net $O \mapsto \cA(O)$ satisfies the \emph{split property}~\cite{DL} if for any pair of double cones $O_1$, $O_2$ such that $\bar O_1 \subset O_2$ (written also $O_1 \subset\subset O_2$), there exists a type I factor $N$ such that $\cA(O_1) \subset N \subset \cA(O_2)$.

\begin{Theorem}\label{thm:isolocal}
Let $\cA$ and $\cB$ be local nets of type $\textup{III}_1$ factors with the split property. Then given any increasing sequence of double cones
\begin{equation}
O_1 \subset\subset O_2 \subset\subset \ldots \subset\subset O_n \subset \subset \ldots
\end{equation}
such that $\bigcup_n O_n = {\mathbb R}^d$, there exists an isomorphism $\bp: \cA \to \cB$ of the quasilocal $C^*$-algebras such that
\begin{equation}
\bp(\cA(O_n)) = \cB(O_n), \quad n=1,2,3,\ldots
\end{equation}
\end{Theorem}

\begin{proof}
By the split property, there is an increasing family $(N_k)_{k\in\bN}$ of type $I$ factors such that
\begin{equation}
\cA(O_1) \subset N_1 \subset \cA(O_2) \subset N_2 \subset \cA(O_3) \subset \ldots \subset N_k \subset \cA(O_{k+1}) \subset \ldots
\end{equation}
so that the quasilocal $C^*$-algebra $\cA$ is the norm-closure of $\bigcup_k N_k$ (``funnel'' of type $I$ factors).
Moreover, consider the universal $C^*$-algebra $O_K^0$ associated to an infinite dimensional separable Hilbert space $K$ and defined as the $C^*$-inductive limit $\lim_k B(K^{\otimes k})$
with respect to the family of inclusions $T \mapsto T\otimes I_K$. Then, as explained in~\cite{Do}, there exists a $C^*$-isomorphism $\phi_\cA$ of $\cA$
onto $O_K^0$,
such that,
for all $k \geq 1$,
$$\phi_\cA(N_k) = B(K^{\otimes k}) = B(K) \otimes \ldots \otimes B(K)\ . $$ In other words, there is a commutative diagram
\begin{equation*}
\begin{array}{ccccccccc}
N_1        & \subset & N_2              & \subset & N_3              & \subset & N_4 & \subset \ldots & \cA \\
\downarrow &         & \downarrow       &         & \downarrow       &         & \downarrow  &  &\downarrow  \\
B(K)       & \hookrightarrow & B(K^{\otimes 2}) & \hookrightarrow & B(K^{\otimes 3}) & \hookrightarrow &  B(K^{\otimes 4}) &\hookrightarrow \ldots & O_K^0
\end{array}
\end{equation*}
where the vertical arrows are surjective $C^*$-isomorphisms.

Now, $\phi_\cA(\cA(O_1)) =: \cA_1$ is a type $\textup{III}_1$ factor in $B(K)$.
Moreover, $\phi_\cA(\cA(O_2))$ is a type $\textup{III}_1$ factor in $B(K \otimes K)$, clearly containing $B(K) \otimes I$.
t is then easy to see that $\phi_\cA(\cA(O_2))$ must be of the form $B(K) \otimes \cA_2$ for a
type $\textup{III}$ factor $\cA_2$ in $B(K)$. Then, by~\cite[Thm.\ 4.22]{SZ}, $\cA_2$ is isomorphic to $\cA(O_2)$, and hence of type $\textup{III}_1$.
By repeating the argument, one must have $\phi_\cA(\cA(O_k))=B(K)\otimes B(K)\otimes \ldots \otimes \cA_k$ ($k$ factors) for suitable type $\textup{III}_1$ factors $\cA_k$ isomorphic to $\cA(O_k)$.

Similarly, there exists an isomorphism $\phi_\cB: \cB \to O_K^0$ such that, for all $k \geq 1$, $\phi_\cB(\cB(O_k))=B(K)\otimes B(K) \otimes \ldots \otimes \cB_k$ for suitable type $\textup{III}_1$ factors $\cB_k$.

Observe that, thanks to the split property, the type $\textup{III}_1$ factors $\cA(O_k)$, $\cB(O_k)$ are hyperfinite and therefore isomorphic~\cite{H}. Then, $\cA_k$ and $\cB_k$ being isomorphic properly infinite von Neumann algebras with properly infinite
commutants (in $B(K)$), for each $k$ there exists a unitary $U_k \in B(K)$ such that $$U_k \cA_k U_k^* = \cB_k \ . $$
One can then construct the ITP-type automorphism $$\alpha := \bigotimes_k {\rm Ad}(U_k)$$ of $O_K^0$ which then clearly satisfies $\alpha \circ \phi_\cA(\cA(O_k)) = \phi_\cB(\cB(O_k))$ for all $k \geq 1$.
The proof is complete by defining $\bp := \phi_\cB^{-1}\circ\alpha\circ\phi_\cA$.
\end{proof}

\begin{Remark}
\begin{remlist}
\item
It is clear from the proof that, more generally, for any two sequences of double cones $O^{(i)}_k$ as above ($i=1,2$), there exists an isomorphism $\bp: \cA \to \cB$ such that $\bp(\cA(O^{(1)}_k))=\cB(O^{(2)}_k)$ for all $k$. Indeed, the same arguments apply even to two nets living on different, possibly curved, spacetimes.
\item
If $O \mapsto \cA(O)$ is the net of von Neumann algebras associated to the massive scalar free field in $d \geq 3$ spacetime dimensions
and $\Aoi$ denotes its unique scaling limit net, due to local normality, a stronger result is available, namely there exists a $C^*$-algebra isomorphism $\bp: \cA \to \Aoi$
between the corresponding quasilocal algebras such that $\bp(\cA(O)) = \Aoi(O)$ for all double cones $O$  based on the hyperplane $x^0=0$~\cite{BV2}.
\end{remlist}
\end{Remark}

\begin{Corollary}
Assume that $\cA$ and its scaling limit $\Aoi$ are nets of type $\textup{III}_1$ factors satisfying the split property.
Then there exists a map $\varphi : \Mor(\Aoi) \to \Mor(\cA,\Aoi)$ such that $\varphi(\rho_0) = \rho_0\bp$, with $\bp : \cA \to \Aoi$ the isomorphism defined in the above theorem.
\end{Corollary}

\section{Functions asymptotically contained in the scaling limit}\label{sec:liftrep}
Given an element of $\Aoi(O)$, we address the problem of lifting it to a suitable function $\l \mapsto A_\l$ taking values in the quasi-local algebra $\cA$ but not necessarily belonging to the scaling algebra $\uAA$. This is connected with the issue of extending sections $\us : \poi(\uAA) \to \uAA$ to the net $\Aoi$ of von Neumann algebras in the scaling limit. As a result, we embed $\uAA$ into a larger net of C$^*$-algebras $\uAAac$ to which $\poi$ extends in such a way that the image of $\uAAac(O)$ contains $\Aoi(O)$.  

Throughout the present section, $\uooi = \lim_\k \uo_{\l_\k}$ will denote a fixed scaling limit state of $\cA$ and $\Aoi$ the corresponding scaling limit net. The following definition already appeared in~\cite{DMV}.

\begin{Definition}\label{def:asymcont}
Let $\l \in \bR_+ \mapsto A_\l \in \cA$ be a norm-bounded function such that $A_\l \in \cA(\l O)$ for some double cone $O$ and all $\l > 0$. Then $\l \mapsto A_\l$ is said to be \emph{asymptotically contained} in $\Aoi(\hat O)$, $\hat O$ a double cone, if for all $\varepsilon > 0$ there exist elements $\uA, \uA' \in \uAA(\hat O)$ such that
\begin{equation*}
\limsup_\k \big(\|(A_{\l_\k} - \uA_{\l_\k})\O\| + \|(A_{\l_\k}^* - \uA'_{\l_\k})\O\| \big) < \varepsilon.
\end{equation*}

We denote by $\uAAac(O)$ the set of functions as above which are asymptotically contained in $\Aoi(\hat O)$ for all $\hat O \supset \bar O$.
When there is no danger of confusion, we will adopt the simplified notation $A$ for a function $\l \mapsto A_\l$ belonging to $\uAAac(O)$.
\end{Definition}

It is clear that $\uAAac(O)$ is a self-adjoint linear space containing $\uAA(O)$ and that the correspondence $O \mapsto \uAAac(O)$ is isotonous, but we can say more.

As in~\cite{DMV}, given $A \in \uAA(O)$ and $h \in C_c(\bR^d)$, we consider elements $\ua_hA \in \uAA(\hat O)$, $\hat O \supset O + \supp h$, defined by
\begin{equation}\label{eq:alphah}
(\ua_h A)_\l := \int_{\bR^d} dx\,h(x)\a_{\l x}(A_\l).
\end{equation}

\begin{Lemma}\label{lem:asymprod}
Let $A, B \in \uAAac(O)$ and let $(h_\n)_{\n \in \bN} \subset C_c(\bR^d)$ be a $\d$-sequence. Then
\begin{equation*}
\lim_{\k,\n} \| [A_{\l_\k} B_{\l_\k} - (\ua_{h_\n}A)_{\l_\k}(\ua_{h_\n}B)_{\l_\k}]\O\| +  \| [A_{\l_\k} B_{\l_\k} - (\ua_{h_\n}A)_{\l_\k}(\ua_{h_\n}B)_{\l_\k}]^*\O\|= 0,
\end{equation*}
where the limit is taken according to the partial (directed) order $(\k',\n')\geq(\k,\n) \iff \k'\geq \k,\n'\geq\n$.
\end{Lemma}

\begin{proof}
According to~\cite[lemma 5.3(c)]{DMV}, there holds
\begin{equation*}
\lim_{\k,\n} \| [A_{\l_\k} - (\ua_{h_\n}A)_{\l_\k}]^*\O \|=0=\lim_{\k,\n} \| [B_{\l_\k} - (\ua_{h_\n}B)_{\l_\k}]\O \|,
\end{equation*}
and therefore, by an $\varepsilon/2$-argument and the boundedness of $\l \mapsto A_\l$, $\l \mapsto B_\l$ in norm, it is sufficient to show that
\begin{equation*}
\lim_{\k,\n} \|[A_{\l_\k} - (\ua_{h_\n} A)_{\l_\k}]B_{\l_\k}\O\| = 0 = \lim_{\k,\n} \|[B_{\l_\k} - (\ua_{h_\n} B)_{\l_\k}]^*A_{\l_\k}^*\O\| .
\end{equation*}
We prove the first equality (the proof of the second one being completely analogous) with a variation of the argument in the proof of the implication (c) $\implies$ (b) of~\cite[lemma 5.3]{DMV}. To this end, let $M = \sup_{\l > 0} \| A_\l \|$, fix $\varepsilon > 0$ and choose correspondingly a $\bar \n \in \bN$ and a $\k^{(1)}_\varepsilon$ such that, for all $\k \geq \k^{(1)}_\varepsilon$,
\begin{equation*}
\|[B_{\l_\k} - (\ua_{h_{\bar \n}}B)_{\l_\k}]\O\| < \varepsilon/8M.
\end{equation*}
Consider then a double cone $\hat O$ containing $O + \supp h_\n$ for all $\n \in \bN$. Then if $W \subset \hat O'$ is a wedge, $\Ooi$ is cyclic for the algebra $\poi(\uAA(W))$ by a standard Reeh-Schlieder argument. Therefore, we can find a $\uC \in \uAA(W)$ such that
\begin{equation*}
\| [\poi(\ua_{h_{\bar \n}}B) - \poi(\uC)]\Ooi \| < \varepsilon/16M,
\end{equation*}
and consequently we can find $\k_\varepsilon^{(2)}$ such that, for $\k \geq \k_\varepsilon^{(2)}$,
\begin{equation*}
\| [(\ua_{h_{\bar \n}}B)_{\l_\k} - \uC_{\l_\k}]\O \| < \varepsilon/8M.
\end{equation*}
Finally, fix $\n_\varepsilon \in \bN$ and $\k_\varepsilon^{(3)}$ such that, if $(\n, \k) \geq (\n_\varepsilon, k_\varepsilon^{(3)})$, there holds
\begin{equation*}
 \|[A_{\l_\k} - (\ua_{h_\n} A)_{\l_\k}]\O\| < \varepsilon/2\|\uC\|.
\end{equation*}
Then if $\n \geq \n_\varepsilon$ and $\k \geq \k_\varepsilon^{(i)}$, $i=1,2,3$, we have the inequalities
\begin{equation*}\begin{split}
 \|[A_{\l_\k} &- (\ua_{h_\n} A)_{\l_\k}]B_{\l_\k}\O\| \\
&\leq  \|[A_{\l_\k} - (\ua_{h_\n} A)_{\l_\k}](B_{\l_\k}-\uC_{\l_\k})\O\|+  \|[A_{\l_\k} - (\ua_{h_\n} A)_{\l_\k}]\uC_{\l_\k}\O\|  \\
&\leq 2M\|(B_{\l_\k}-\uC_{\l_\k})\O\| + \|\uC_{\l_\k}[A_{\l_\k} - (\ua_{h_\n} A)_{\l_\k}]\O\| \\
&\leq 2M\big(\|[B_{\l_\k} - (\ua_{h_{\bar \n}}B)_{\l_\k}]\O\| + \| [(\ua_{h_{\bar \n}}B)_{\l_\k} - \uC_{\l_\k}]\O \|\big)\\
&\quad+\|\uC\| \|[A_{\l_\k} - (\ua_{h_\n} A)_{\l_\k}]\O\| < \varepsilon,
\end{split}\end{equation*}
where in the second inequality we have taken into account the fact that $\uC_\l$ commutes with $A_{\l} - (\ua_{h_\n} A)_{\l}$.
\end{proof}

\begin{Proposition}
$\uAAac(O)$ is a C$^*$-subalgebra of $B(\bR_+, \cA)$ for every double cone $O$.
\end{Proposition}

\begin{proof}
We start by showing that $\uAAac(O)$ is a *-algebra. Given $A,B \in \uAAac(O)$ and $\hat O \supset \bar O$, one has $\ua_{h_\n}A\, \ua_{h_\n} B \in \uAA(\hat O)$ for $\n$ sufficiently large, and according to the above lemma,
\begin{equation*}
\limsup_\k \| [A_{\l_\k} B_{\l_\k} - (\ua_{h_\n}A)_{\l_\k}(\ua_{h_\n}B)_{\l_\k}]\O\| +  \| [A_{\l_\k} B_{\l_\k} - (\ua_{h_\n}A)_{\l_\k}(\ua_{h_\n}B)_{\l_\k}]^*\O\|
\end{equation*}
can be made arbitrarily small by taking $\n$ sufficiently large.

It is then sufficient to show that $\uAAac(O)$ is norm-closed in $B(\bR_+,\cA)$. To this end, consider a sequence $(A_n)_{n \in \bN} \subset \uAAac(O)$ and $A \in B(\bR_+,\cA)$ such that
\begin{equation*}
\lim_{n \to +\infty} \sup_\l \| A_{n,\l}-A_\l \| = 0.
\end{equation*}
Since $A_{n,\l} \in \cA(\l O)$ for all $n \in \bN$, $\l > 0$, this implies $A_\l \in \cA(\l O)$ for all $\l > 0$. Given now $\varepsilon > 0$, fix $n \in \bN$ such that $\sup_\l \| A_{n,\l}-A_\l \| < \varepsilon/4$. Then, for any $\hat O \supset \bar O$, we can choose $\uA_n, \uA'_n \in \uAA(\hat O)$ such that
\begin{equation*}
\limsup_\k \big(\|(A_{n,\l_\k} - \uA_{n,\l_\k})\O\| + \|(A_{n,\l_\k}^* - \uA'_{n,\l_\k})\O\| \big) < \varepsilon/2.
\end{equation*}
Then by an $\varepsilon/2$-argument,
\begin{equation*}
\limsup_\k \big(\|(A_{\l_\k} - \uA_{n,\l_\k})\O\| + \|(A_{\l_\k}^* - \uA'_{n,\l_\k})\O\| \big) < \varepsilon,
\end{equation*}
which proves the statement.
\end{proof}

It is obvious that $O \mapsto \uAAac(O)$ is a local net of C$^*$-algebras.

Next, we show that it is possible to extend the scaling limit representation $\poi$ from $\uAA$ to the inductive limit $\uAAac \subset B(\bR_+,\cA)$
of the C$^*$-algebras $\uAAac(O)$. We begin with some preparatory lemmas.

\begin{Lemma}\label{lem:poiextend}
There exists a unique bounded self-adjoint linear map $\poiac: \uAAac \to B(\Hoi)$ which extends $\poi:\uAA \to \Aoi$ and such that, for all $A \in \uAAac$,
\begin{equation}\label{eq:poiextend}
\langle \poi(\uB)\Ooi, \poiac(A)\poi(\uC)\Ooi\rangle = \lim_\k \o(\uB_{\l_\k}^*A_{\l_\k}\uC_{\l_\k}), \quad \uB, \uC \in \uAA.
\end{equation}
Furthermore, $\poiac(\uAAac(O)) \subset \Aoi^d(O) := \Aoi(O')'$ for all double cones $O$.
\end{Lemma}

\begin{proof}
The limit on the right hand side of~\eqref{eq:poiextend} exists because, as shown in~\cite{BDM}, it can be viewed as the evaluation,  on the function $\l \mapsto \o(\uB_\l^*A_\l \uC_\l)$, of a multiplicative mean on the algebra of bounded functions on $\bR_+$. Furthermore, there holds the bound
\begin{equation*}\begin{split}
\big|\lim_\k\o(\uB_{\l_\k}^*A_{\l_\k}\uC_{\l_\k})\big| &\leq \sup_\l \| A_\l\| \lim_\k \| \uB_{\l_\k}\O\| \| \uC_{\l_\k}\O\| \\
&=\sup_\l \| A_\l \| \| \poi(\uB)\Ooi\| \|\poi(\uC)\Ooi\|.
\end{split}\end{equation*}
Thanks to the density of $\poi(\uAA)\Ooi$ in $\Hoi$, this shows the existence of a unique bounded operator $\poiac(A)$ on $\Hoi$ satisfying~\eqref{eq:poiextend}.
The map $A \mapsto \poiac(A)$ is clearly linear, self-adjoint, with norm bounded by 1, and extends $\poi$ since for $\uA \in \uAA$,
\begin{equation*}
\lim_\k \o(\uB_{\l_\k}^*\uA_{\l_\k}\uC_{\l_\k})= \langle\poi(\uB)\Ooi, \poi(\uA)\poi(\uC)\Ooi\rangle.
\end{equation*}
Finally, if $A \in \uAAac(O)$, $O_1 \subset O'$ and $\underline{D} \in \uAA(O_1)$, there holds
\begin{equation*}\begin{split}
\langle \poi(\uB)&\Ooi, \poiac(A)\poi(\underline{D})\poi(\uC)\Ooi\rangle = \lim_\k \o(\uB_{\l_\k}^*A_{\l_\k}\underline{D}_{\l_\k}\uC_{\l_\k}) \\
&= \lim_\k \o(\uB_{\l_\k}^*\underline{D}_{\l_\k}A_{\l_\k}\uC_{\l_\k}) = \langle \poi(\uB)\Ooi, \poi(\underline{D})\poiac(A)\poi(\uC)\Ooi\rangle,
\end{split}\end{equation*}
where in the second equality we have used the fact that $A_\l \in \cA(\l O)$ for all $\l > 0$. This shows that $\poiac(A) \in \Aoi(O')'$.
\end{proof}

\begin{Lemma}\label{lem:poiextendlim}
For each $\d$-sequence $(h_\nu)_{\nu \in \bN}$ and for each $A \in \uAAac(O)$, there holds
\begin{equation}\label{eq:poiextendlim}
\sslim_{\nu \to +\infty}\poi(\ua_{h_\nu}A) = \poiac(A)
\end{equation}
(limit in the strong* operator topology).
\end{Lemma}

\begin{proof}
By the Reeh-Schlieder property of wedges in the scaling limit, it is sufficient to show that the equations
\begin{equation}\label{eq:poiextendwlim}
\lim_{\nu \to +\infty} \big\| \big[\poiac(A)-\poi(\ua_{h_n}A)\big]\poi(\uC)\Ooi\big\| = 0 = \lim_{\nu \to +\infty} \big\| \big[\poiac(A)-\poi(\ua_{h_n}A)\big]^*\poi(\uC)\Ooi\big\|
\end{equation}
hold for all $\uC \in \uAA(W)$, where $W$ is a wedge spacelike to some double cone $\hat O \supset \bar O$. As the proof is the same, we only show the validity of the first equation. Fix then $\varepsilon > 0$ and, thanks to~\cite[Lemma 5.3(c)]{DMV}, choose $\n_\varepsilon, \k_\varepsilon$ such that for $\n \geq \n_\varepsilon$, $\k \geq \k_\varepsilon$ one has
\begin{equation*}
\| [A_{\l_\k}- (\ua_{h_\n}A)_{\l_\k}]\O \| < \varepsilon.
\end{equation*}
This implies that for $\n \geq \n_\varepsilon$
\begin{equation*}
 \lim_\k \big\|\uC_{\l_\k}[A_{\l_\k}- (\ua_{h_\n}A)_{\l_\k}]\O\big\| \leq \|\uC\|\varepsilon.
\end{equation*}
Keeping then in mind that, according to the previous Lemma and to the localization properties of $\ua_{h_\n}A$, $\poiac(A)-\poi(\ua_{h_\n}A)$ commutes with $\poi(\uC)$, we see that equation~\eqref{eq:poiextendwlim}  holds.
\end{proof}

We denote by  $\Aoir$ the inductive limit of the outer regularized von Neumann algebras
\begin{equation*}
\Aoir(O) := \bigcap_{\hat O \supset \bar O} \Aoi(\hat O).
\end{equation*}
We draw the conclusions of the above discussion in the following form.

\begin{Theorem}\label{thm:poiextend}
There exists a morphism $\poiac : \uAAac \to \Aoi$ which extends $\poi$.
 \end{Theorem}

\begin{proof}
It follows from Lemmas~\ref{lem:poiextend}, \ref{lem:poiextendlim} that for all double cones $O$, $\poiac(\uAAac(O)) \subset \Aoir(O) \cap \Aoi^d(O)$. Then, since the quasi-local algebras $\Aoir$ and $\Aoi \subset \Aoi^d$ coincide, one obtains, by norm continuity, a bounded self-adjoint linear map $\poiac :\uAAac \to \Aoi$.

In order to show that $\poiac$ is multiplicative we can argue as in the proof of Lemma~\ref{lem:poiextendlim}, using Lemma~\ref{lem:asymprod} to show that
\begin{equation*}
\lim_{\n \to +\infty} \big\| \big[\poiac(AB)-\poi(\ua_{h_\n}A)\poi(\ua_{h_\n}B)\big]\poi(\uC)\Ooi\big\|=0
\end{equation*}
holds for all $A, B \in \uAAac(O)$ and all $\uC \in \uAA(W)$ with $W$ a wedge spacelike to some $\hat O \supset \bar O$.
\end{proof}

In the remaining part of this section, we show that, under fairly general conditions, the local von Neumann
algebras $\Aoi(O)$ of the scaling limit theory are in the image of $\poiac$.

\begin{Lemma}\label{lem:step}
Let $(\uA_n)_{n \in \bN} \subset \uAA(O)$ be a norm-bounded sequence and $(\l_n)_{n \in \bN} \subset \bR_+$ be such that $\l_n \searrow 0$. Assume in addition that
\begin{equation}\label{eq:stepinequal}
\| (\uA_{n+1,\l}-\uA_{n,\l})\O \| < \frac{1}{2^n}, \quad \| (\uA_{n+1,\l}-\uA_{n,\l})^*\O \| < \frac{1}{2^n},\quad \l \leq \l_n, n \in \bN.
\end{equation}
Then the function $\l \mapsto A_\l$ defined by $A_\l = \uA_{n,\l}$ for $\l \in (\l_{n+1},\l_n]$ is asymptotically contained in $\Aoi(O)$.
\end{Lemma}

\begin{proof}
If $\l < \l_n$ there is an $m \geq n$ such that $\l \in (\l_{m+1}, \l_m]$ and therefore
\begin{equation*}\begin{split}
\| (A_\l - \uA_{n,\l})\O\| &= \| (\uA_{m,\l}-\uA_{n,\l})\O\| \\
&\leq \sum_{k=n}^{m-1}\| (\uA_{k+1,\l}-\uA_{k,\l})\O \| < \frac{1}{2^{n-1}},
\end{split}\end{equation*}
since $\l \leq \l_m < \l_{m-1} < \dots < \l_n$. Similarly $\sup_{\l < \l_n}\| (A_\l - \uA_{n,\l})^*\O\| < 1/2^{n-1}$, and then
\begin{equation*}
\limsup_{\l \to 0} \big(\| (A_\l - \uA_{n,\l})\O\| +\| (A_\l - \uA_{n,\l})^*\O\| \big) < \frac{1}{2^{n-2}},
\end{equation*}
which shows that $\l \mapsto A_\l$ is asymptotically contained in $\Aoi(O)$, keeping in mind the fact that for all $\l > 0$ there exists a $\kappa$ such that $\kappa' \geq \kappa$ implies $\l_{\kappa'} \leq \l$.
\end{proof}

In the following theorem we use the notion of convergent scaling limit as in~\cite[Def. 4.4]{BDM2}. We recall that this means that there is a C$^*$-subalgebra $\uAA_\text{conv} \subset \uAA$ such that $\poi(\uAA_\text{conv}\cap\uAA(O))'' = \Aoi(O)$ for every $O$ and, for all $\uA \in \uAA_\text{conv}$, there exists $\lim_{\l \to 0} \o(\uA_\l) = \uooi(\uA)$. This condition is expected to hold in many physically interesting models and can be explicitly verified, e.g., in dilation covariant theories satisfying Haag-Swieca compactness and in the theory of the Klein-Gordon field in $d =3, 4$ spacetime dimensions~\cite[Thms.~7.1, 7.5]{BDM2}.

\begin{Theorem}\label{thm:steplimit}
Assume that $\cA$ has a convergent scaling limit. Then $$\Aoi(O) \subset \poiac(\uAAac(O))$$ for all double cones $O$.
\end{Theorem}

\begin{proof}
It is sufficient to show that, given $A_0 \in \Aoi(O)$, we can choose an element $A$ asymptotically
contained in $\Aoi(O)$ such that 
\begin{equation}\label{eq:steplimit}
\lim_{\n \to +\infty}\poi(\ua_{h_\n}A) = A_0
\end{equation}
in the strong* operator topology, for every $\d$-sequence $(h_\n)_{\n \in \bN}$. In fact, if this is true,
the conclusion is obtained by observing that an element asymptotically contained in $\Aoi(O)$ is also asymptotically contained in $\Aoi(\hat O)$ for all $\hat O \supset \bar O$, and therefore we have $A \in \uAAac(O)$ and $\poiac(A) = A_0$ according to Lemma~\ref{lem:poiextendlim}.

Consider then $A_0 \in \Aoi(O)$. Since $\poi(\uAA_\text{conv}\cap\uAA(O))$ is strongly* dense in $\Aoi(O)$, thanks to Kaplanski's density theorem and to the fact that $\Ooi$ is separating for $\Aoi(O)$ there are elements $\uA_n \in \uAA_\text{conv}\cap\uAA(O)$ such that $\| \uA_n \| \leq \| A_0\|$ and
\begin{equation*}
\lim_{n \to +\infty}\poi(\uA_n) = A_0
\end{equation*}
in the strong* operator topology. Passing to a subsequence, if necessary, we can also assume that
\begin{equation*}
\| (\poi(\uA_{n+1}) - \poi(\uA_n))\Ooi \| < \frac{1}{2^n}, \quad\| (\poi(\uA_{n+1}) - \poi(\uA_n))^*\Ooi \| < \frac{1}{2^n}.
\end{equation*}
Now, since for $\uB \in \uAA_\text{conv}$ one has $\| \poi(\uB)\Ooi \| = \lim_{\l \to 0} \| \uB_\l \O \|$, we see that we can find a sequence $(\l_n)_{n \in \bN} \subset \bR_+$ such that $\l_n \searrow 0$ and \eqref{eq:stepinequal} holds. Therefore the function $\l \mapsto A_\l$ defined as in Lemma~\ref{lem:step} is asymptotically contained in $\Aoi(O)$ and satisfies $\sup_\l \| A_\l \| \leq \| A_0 \|$, so we only need to show that~\eqref{eq:steplimit} holds.

To this end, given $\varepsilon > 0$, we fix $n \in \bN$ such that
\begin{gather*}
\| (\poi(\uA_n) - A_0)\Ooi\| < \varepsilon, \\
\limsup_{\l \to 0} \big(\| (A_\l - \uA_{n,\l})\O\| +\| (A_\l - \uA_{n,\l})^*\O\| \big) <\varepsilon.
\end{gather*}
Correspondingly, we can find $\n_\varepsilon$ such that if $\n > \n_\varepsilon$
\begin{equation*}
\int_{\bR^d} dx\,h_\n(x) \| \ua_x(\uA_n)-\uA_n \| < \varepsilon.
\end{equation*}
The two latter inequalities imply that for sufficiently small $\l$ and $\n > \n_\varepsilon$
\begin{equation*}\begin{split}
\big\| \big((\ua_{h_\n}A)_\l - \uA_{n,\l}\big)\O \big\|  &\leq \left\| \left(\int_{\bR^d} dx\,h_\n(x)\a_{\l x}(A_\l) - \int_{\bR^d} dx\,h_\n(x)\a_{\l x}(\uA_{n,\l})\right)\O\right\| \\
&\quad+ \left\|\int_{\bR^d} dx\,h_\n(x)\big(\a_{\l x}(\uA_{n,\l})-\uA_{n,\l}\big)\O\right\| \\
&\leq \|(A_\l - \uA_{n,\l})\O\| + \int_{\bR^d} dx\,h_\n(x) \| \ua_x(\uA_n)-\uA_n \| < 2\varepsilon,
\end{split}\end{equation*}
and therefore, for $\n > \n_\varepsilon$,
\begin{equation*}
\| (\poi(\ua_{h_\n}A) - \poi(\uA_n))\Ooi\| = \lim_\k \big\| \big((\ua_{h_\n}A)_{\l_\k} - \uA_{n,{\l_\k}}\big)\O \big\| < 2\varepsilon.
\end{equation*} Finally, the inequality
\begin{equation*}\begin{split}
\|(\poi(\ua_{h_\n}A) &- A_0)\Ooi\| \leq \|(\poi(\ua_{h_\n}A) - \poi(\uA_n))\Ooi\|+ \| (\poi(\uA_n) - A_0)\Ooi\| < 3\varepsilon
\end{split}
\end{equation*}
and a similar argument for $\|(\poi(\ua_{h_\n}A) - A_0)^*\Ooi\|$ show that~\eqref{eq:steplimit} holds. 
\end{proof}

It seems of interest to decide whether the above argument can be generalized by dropping the assumption of convergent scaling limit for $\cA$.

\section{Asymptotic morphisms and scaling limit morphisms}\label{sec:asympmor}

Thanks to the results of the previous section, we are now ready to introduce a natural notion of asymptotic morphism adapted to the situation at hand, following the analogy discussed in Sec.~\ref{sec:dilation}. The present section will be devoted to the study of the relations
between these objects and the morphisms from $\cA$ or $\Aoi$ into $\Aoi$.

\begin{Definition}\label{def:tameasympmor}
By an \emph{asymptotic morphism} of the net $\cA$ relative to the scaling limit state $\uooi = \lim_\k \uo_{\l_\k}$ we mean a family $(\r_\l)_{\l > 0}$ of maps $\rho_\l : \cA \to \cA$ such that, for all $A,B \in \cA$, $\a \in \bC$:
\begin{align}
&\lim_{\k} \| [\rho_{\l_\k}(A)^*-\rho_{\l_\k}(A^*)]\O\|= 0, \label{eq:asympadjoint}\\
&\lim_\k \| [\rho_{\l_\k}(A+\a B) - \rho_{\l_\k}(A) - \a \rho_{\l_\k}(B)]\O\|=0, \label{eq:asymplinear}\\
&\lim_\k \| [\rho_{\l_\k}(AB) - \rho_{\l_\k}(A)\rho_{\l_\k}(B)]\O\|=0. \label{eq:asympmult}
\end{align}
The asymptotic morphism $(\r_\l)_{\l>0}$ will be called \emph{tame} if the following properties hold:
\smallskip
\begin{proplist}{3}
\item for each $A \in \cA$ the map $\rac(A) : \bR_+ \to \cA$ defined by $\l \mapsto \rho_\l(A)$ is an element of $\uAAac$;
\item the map $A \in \cA \mapsto \rac(A) \in \uAAac$ is norm-continuous;
\item $\poiac(\rac(\bigcup_O \cA(O))) \subset \bigcup_O \Aoi(O)$.
\end{proplist}
\end{Definition}

Note that a morphism $\r : \cA \to \cA$ naturally defines an asymptotic morphism according to the above definition, setting $\r_\l := \r$ for all $\l > 0$; however $(\r_\l)$ thus defined will not be tame in general.

We will denote by $\AMor_\iota(\cA)$ the set of tame asymptotic morphisms of $\cA$ relative to the scaling
limit state $\uooi$. When there is no possibility of confusion, we will however often drop the subscript $\iota$ in order to have a lighter notation.

We will also say that a morphism $\br: \cA \to \cB$, with $\cA, \cB$ nets of C$^*$-algebras, is \emph{properly supported} if $\br(\bigcup_O \cA(O)) \subset \bigcup_O \cB(O)$. The set of properly supported morphisms from $\cA$ to $\cB$ will be denoted by $\Morps(\cA,\cB)$. If $\bp:\cA \to \cB$ is an isomorphism, we will call it properly supported, by a slight abuse of terminology, if it is properly supported in the above sense along with its inverse or, equivalently, if $\bp\big(\bigcup_O \cA(O)\big) = \bigcup_O \cB(O)$.
\medskip

A way to produce asymptotic morphisms is to ``lift'' morphisms $\cA \to \Aoi$ as explained in the
following result.

\begin{Theorem}\label{thm:asympmor}
Assume that $\cA$ has a convergent scaling limit, and let $\br :  \cA \to \Aoi$ be a properly supported morphism. Then there exists a tame asymptotic morphism $(\rho_\l)_{\l>0}$ of $\cA$ such that 
\begin{equation}\label{eq:asympmor}
\poiac(\rac(A)) = \br(A), \qquad A \in \cA.
\end{equation}
\end{Theorem}

\begin{proof}
Since $\Aoi = \poiac(\uAAac)$, Thm.~\ref{thm:steplimit}, thanks to the Bartle-Graves Selection Theorem~\cite[TVS II.35, Prop.12]{Bourbaki?} we can choose a continuous section of $\poiac: \uAAac \to \Aoi$, i.e., a norm-continuous map
$s : \Aoi \to \uAAac$ such that $\poiac\circ s = \Id_{\Aoi}$. We define then $\rho_\l(A) := s(\br(A))_\l$ for all $\l > 0$, $A \in \cA$. Properties (i)-(iii) of Definition~\ref{def:tameasympmor} and Eq.~\eqref{eq:asympmor} follow at once. We show that
\begin{equation*}
\lim_\k \big\|\big[\rho_{\l_\k}(AB) -\rho_{\l_\k}(A)\rho_{\l_\k}(B)\big]\O\big\|  = 0, \qquad A, B \in \cA.
\end{equation*}
Since $\uAAac$ is a C$^*$-algebra, setting
\begin{equation*}
Z_\l := s(\br(AB))_\l - s(\br(A))_\l s(\br(B))_\l, \qquad \l > 0,
\end{equation*}
we have $Z \in \uAAac$, and therefore, using~\eqref{eq:poiextend} with $\uB = \uC = \Id$, and the fact that, by Theorem~\ref{thm:poiextend}, $\poiac$ is a morphism,
\begin{equation*}\begin{split}
\lim_\k \big\|\big[\rho_{\l_\k}(AB) &-\rho_{\l_\k}(A)\rho_{\l_\k}(B)\big]\O\big\|^2  = \lim_\k \o(Z_{\l_\k}^*Z_{\l_\k}) \\
&= \ooi(\poiac(Z^*Z)) \\
&= \big\| \big[\poiac(s(\br(AB)))-\poiac(s(\br(A)))\poiac(s(\br(B)))\big]\Ooi\big\|^2 \\
&= \big\|\big[\br(AB)-\br(A)\br(B)\big]\Ooi\|^2 = 0.
\end{split}\end{equation*}
The proof of the equations $\lim_\k \|[\rho_{\l_k}(A^*)-\rho_{\l_\k}(A)^*]\O\|=0 =\lim_\k \|[\rho_{\l_\k}(A + \a B)-\r_{\l_\k}(A)-\a\r_{\l_\k}(B)]\O\|=0$ is analogous, and this shows indeed that $(\r_\l)$ is an asymptotic morphism of $\cA$.
\end{proof}

\begin{Remarks}
\begin{remlist}
\item The assumption $\br(\cup_O\cA(O)) \subset \cup_O\Aoi(O)$, which is only needed to show the validity of property (iii) of Definition~\ref{def:tameasympmor}, is natural having in mind applications to superselection theory: e.g.\ it is satisfied by $\br := \r_0\bp$, where $\r_0$ is a localized morphims of $\Aoi$ and $\bp : \cA \to \Aoi$ is an isomorphism as in 
Theorem~\ref{thm:isolocal}.
\item In the setting of the Theorem, it is not granted that for each $A \in \cup_O\cA(O)$, there exists a double cone $O_1$ with $\r_\l(A) \in \cA(\l O_1)$ for all $\l>0$. The validity of such property can actually be obtained, by a suitable modification of the above argument, at the price of restricting the domain of the asymptotic morphism $(\rho_\l)_\l$ to the dense $*$-subalgebra $\cup_O \cA(O)$. 
\end{remlist}
\end{Remarks}

The previous theorem has the following converse.

\begin{Theorem}\label{thm:asympmortomor}
Let $(\rho_\l)$ be a tame asymptotic morphism of a local net $\cA$ relative to $\uooi$. Then defining $\br(A) := \poiac(\rac(A))$, $A \in \cA$, we get a properly supported morphism $\br : \cA \to \Aoi$.
\end{Theorem}

\begin{proof}
Consider $A, B \in \bigcup_O \cA(O)$. Then we have
\begin{equation*}\begin{split}
\| [\br(AB)-\br(A)\br(B)]\Ooi\| &= \| \poi(\rac(AB)-\rac(A)\rac(B))\Ooi\| \\
&= \lim_\k \| [\r_{\l_\k}(AB)-\r_{\l_\k}(A)\r_{\l_\k}(B)]\O\| = 0.
\end{split}\end{equation*}
Thanks to the fact that $\Ooi$ is separating for local algebras and to property (iii) of Definition~\ref{def:tameasympmor}, this implies that $\br$ is multiplicative on strictly local elements. The multiplicativity of $\br$ on $\cA$ then follows from the fact that, according to property (ii) of Definition~\ref{def:tameasympmor}, it is norm-continuous. The proof of the linearity and self-adjointness of $\br$ follows the same lines.
\end{proof}

Obviously, given morphisms $\br, \bp : \cA \to \Aoi$ with $\bp$ an isomorphism, we get a morphism $\r_0$ of $\Aoi$ simply by $\r_0 := \br \bp^{-1}$, and all morphisms of $\Aoi$ can be obtained in this way. In view of the
correspondence between morphisms $\cA \to \Aoi$ and asymptotic morphisms of $\cA$ established
above, it is then possible to describe morphisms of $\Aoi$ by pairs of suitable asymptotic morphisms.

\begin{Definition}\label{def:asympiso}
An \emph{asymptotic isomorphism} of $\cA$ relative to the scaling limit state $\uooi$ is a tame asymptotic morphism $(\phi_\l)$ of $\cA$ relative to $\uooi$ such that:
\smallskip
\begin{proplist}{2}
\item $\pac:\cA \to \uAAac$ is injective;
\item there exists a continuous section $\bar s : \Aoi \to \uAAac$ of $\poiac$ such that 
$$\pac\big(\bigcup_O\cA(O)\big) = \bar s\big(\bigcup_O\Aoi(O)\big).$$ 
\end{proplist}
\end{Definition}

We will use the notation $\AIso_\iota(\cA)$  (or simply $\AIso(\cA)$) to denote the set of asymptotic isomorphisms of $\cA$ relative to $\uooi$.

\begin{Lemma}\label{lem:asympiso}
Assume that $\cA$ has a convergent scaling limit. If $(\phi_\l)$ is an asymptotic isomorphism of $\cA$, then the morphism $\bp : \cA \to \Aoi$ associated to it according to Theorem~\ref{thm:asympmortomor} is a properly supported isomorphism. Conversely given a properly supported isomorphism $\bp : \cA \to \Aoi$, the asymptotic morphism obtained in Theorem~\ref{thm:asympmor} can be chosen to be an asymptotic isomorphism.
\end{Lemma}

\begin{proof}
Let $(\phi_\l)$ be an asymptotic isomorphism of $\cA$. We show that the associated morphism $\bp : \cA \to \Aoi$ is injective. Since $\pac$ is injective and $\bp = \poiac\pac$, it is sufficient to show that $\pac(\cA)\cap \ker \poiac = \{0\}$.
In order to do that, we first note that, if $A \in \bigcup_O \cA(O)$, by property (ii) of Def.~\ref{def:asympiso} there is $A_0 \in \bigcup_O \Aoi(O)$ such that $\pac(A) = \bar s (A_0)$, and applying $\poiac$ to both sides, one sees that $\pac(A) = \bar s (\poiac\pac(A))$. By continuity of $\pac, \bar s$ and $\poiac$, this last equation is actually valid for all $A \in \cA$, which entails $\pac(\cA) \subset \bar s (\Aoi)$. Then one has
$$\pac(\cA)\cap\ker\poiac \subset \bar s(\Aoi)\cap\ker\poiac = \{ 0 \},$$
since $\poiac \bar s = \Id$.
The surjectivity of $\bp$ follows from property (ii) of Def.~\ref{def:asympiso} and continuity since
$$\bp\big(\bigcup_O\cA(O)\big) = \poiac\big(\pac\big(\bigcup_O\cA(O)\big)\big) = \poiac\big(\bar s\big(\bigcup_O\Aoi(O)\big)\big) = \bigcup_O\Aoi(O).$$
The last equation also shows that $\bp$ is a properly supported isomorphism.

Assume now that $\bp : \cA \to \Aoi$ is a properly supported isomorphism. As seen in the proof of Theorem~\ref{thm:asympmor}, there exists a tame asymptotic morphism $(\phi_\l)$ such that $\pac(A) = \bar s (\bp(A))$, $A \in \cA$, for some continuous section $\bar s : \Aoi \to \uAAac$ of $\poiac$. This immediately entails property (ii) of Def.~\ref{def:asympiso} for $\pac$. Furthermore, if $\pac(A) = \pac(B)$, $A, B \in \cA$, then applying $\poiac$ to both members one gets $\bp(A) = \bp(B)$, and therefore $A = B$, $\bp$ being an isomorphism.
\end{proof}

It is well known that under fairly general assumptions  the quasi-local C$^*$-algebra $\cA$ turns out to be simple, and thus in this case a nonzero morphism from $\cA$ into $\Aoi$ is automatically injective.

We are now in the position of showing that, under suitable hypotheses, morphisms of $\Aoi$ can be described by means of asymptotic morphisms of $\cA$.

\begin{Theorem}\label{thm:difference}
Assume that $\cA$ has a convergent scaling limit and that there exists a properly supported isomorphism $\bp:\cA\to\Aoi$. Given a properly supported morphism $\r_0$ of $\Aoi$, there exists a tame asymptotic morphism $(\r_\l)$ and an asymptotic isomorphism $(\phi_\l)$ of $\cA$ such that
\begin{equation}\label{eq:difference}
\r_0 = \poiac \rac (\pac)^{-1}\bar s,
\end{equation}
where $\bar s$ is the section in Definition~\ref{def:asympiso}.(ii).

Conversely, given a tame asymptotic morphism $(\r_\l)$ and an asymptotic isomorphism $(\phi_\l)$ of $\cA$,
$\r_0$ defined by formula~\eqref{eq:difference} is a properly supported morphism of $\Aoi$.
\end{Theorem}

\begin{proof}
We start by noticing that $\br := \r_0 \bp : \cA \to \Aoi$ is a properly supported morphism. Consider then a tame asymptotic morphism $(\r_\l)$ associated to $\br$ according to Theorem~\ref{thm:asympmor} and an asymptotic isomorphism $(\phi_\l)$ associated to $\bp$ according to Lemma~\ref{lem:asympiso}. Since $\pac$ is injective and $\bar s(\Aoi) = \pac(\cA)$, there holds $\bp^{-1} = (\pac)^{-1}\bar s$, and therefore
\begin{equation*}
\r_0 = \br \bp^{-1} = \poiac \rac (\pac)^{-1}\bar s.
\end{equation*}
The converse statement is clear since $\br = \poiac \rac : \cA \to \Aoi$ and $\bp^{-1} = (\pac)^{-1}\bar s : \Aoi \to \cA$ are both properly supported morphisms.
\end{proof}

\begin{Remarks}
\begin{remlist}
\item Loosely speaking, every morphism of $\Aoi$ is thus realized as a (multiplicative) difference of asymptotic morphisms. This is somewhat reminiscent of the representation of KK-elements in the Cuntz picture. 
\item If we had a canonical isomorphism between $\cA$ and $\Aoi$ (as for dilation covariant theories), we could identify morphisms of $\Aoi$ with morphisms between $\cA$ and $\Aoi$ and thus we could avoid talking about differences. In the general case, the usage of differences allows to bypass the arbitrariness of the choice of a noncanonical isomorphism.
\end{remlist}
\end{Remarks}

Theorem~\ref{thm:difference} allows us to define a surjective map 
\begin{equation}\label{eq:Psi}
\Psi : \AMor_\iota(\cA) \times \AIso_\iota(\cA) \to \Morps(\Aoi)
\end{equation}
through formula~\eqref{eq:difference}. Notice that the second part of Theorem~\ref{thm:difference} is actually independent of the assumption of convergent scaling limit for $\cA$. Therefore, also the definition of $\Psi$ does not rely on this hypothesis.

Next, we turn to consider equivalence of pairs of asymptotic morphisms. We will give below such a notion
so that equivalent pairs of asymptotic morphisms are mapped
to unitarily equivalent morphisms by $\Psi$.

As a start, it is useful to consider a few relevant notions of equivalence for (single) asymptotic morphisms.

We recall that two morphisms $\r,\s : \cB_1 \to \cB_2$, $\cB_1$, $\cB_2$ C$^*$-algebras,  are unitarily equivalent if there exists
a unitary operator $U \in \cB_2$ such that $\r = U \s(\cdot)U^*$.

We say that $U = (U_\l) \in \uAAac$ is an asymptotic unitary if $\poiac(U) \in \Aoi$ is unitary.
It is easy to see that if $U \in \uAAac(W)$, with $W$ a wedge, this is equivalent to
\begin{equation*}
\lim_\k \| [U_{\l_\k}U_{\l_\k}^*-1]\O \| = 0 = \lim_\k \|[U_{\l_\k}^*U_{\l_\k}-1]\O \|. \end{equation*}

\begin{Definition}\label{def:asympmorequiv}
Let $(\r_\l)$, $(\s_\l)$ be asymptotic morphisms of $\cA$ relative to $\uooi$. We say that:

\begin{proplist}{2}
\item $(\r_\l)$, $(\s_\l)$ are \emph{asymptotically equal} if, for all $A \in \cA$,
\begin{equation*}
\lim_\k \| [\r_{\l_\k}(A) -\s_{\l_\k}(A)]\O \| = 0; 
\end{equation*}
\item $(\r_\l)$, $(\s_\l)$ are \emph{asymptotically (inner) equivalent} if there exists
an asymptotic unitary $(U_\l) \in \uAAac$ such that $\r'_\l := U_\l \r_\l(\cdot) U_\l^*$ defines an asymptotic
morphism asymptotically equal to $(\s_\l)$.
\end{proplist}
\end{Definition}

If $(\r_\l)$, $(\s_\l)$ are tame asymptotic morphisms, it is easy to see that if they are asymptotically equal
then the associated morphisms $\br, \bs : \cA \to \Aoi$ coincide, and conversely that given a properly supported morphism $\br :\cA \to \Aoi$ all the tame asymptotic morphisms $(\r_\l)$ of $\cA$ such that $\br = \poiac\rac$ are asymptotically
equal to each other. Similarly, $(\r_\l)$, $(\s_\l)$ are asymptotically equivalent if and only if $\br, \bs$ are unitarily equivalent through a unitary $U_0  = \poiac(U) \in \Aoi$. This also shows that asymptotic equivalence is actually an equivalence relation. Note that in general $\r'_\l = U_\l \r_\l(\cdot) U_\l^*$ will
not be a tame asymptotic morphism, however it makes perfect sense to consider $\poiac(U \rac(\cdot) U^*) = \poiac(U) \br(\cdot) \poiac(U)^*$ which is a morphism from $\cA$ to $\Aoi$.

We also remark that there is a natural right action of $\aut(\cA)$ on the set of asymptotic morphisms of $\cA$
given by associating to $\a \in \aut(\cA)$ and $(\r_\l)$ the asymptotic morphism $(\r_\l\a)$. Moreover,
if $\a(\bigcup_O \cA(O)) \subset \bigcup_O\cA(O)$ we also have that $(\r_\l\a) \in \AMor(\cA)$ for all
$(\r_\l)\in\AMor(\cA)$, and the morphism associated to $(\r_\l \a)$ as in Theorem~\ref{thm:asympmortomor}
is $\br\a$. Furthermore if $\a(\bigcup_O \cA(O)) = \bigcup_O\cA(O)$ the action of $\a$ preserves $\AIso(\cA)$.

\begin{Definition}\label{def:pairasympmorequiv}
We say that $\big((\r_\l),(\phi_\l)\big), \big((\s_\l),(\psi_\l)\big) \in \AMor_\iota(\cA) \times \AIso_\iota(\cA)$ are \emph{asymptotically
equivalent} if $(\r_\l \bp^{-1}\bpsi)$  and $(\s_\l)$ are asymptotically equivalent as asymptotic morphisms (see Definition~\ref{def:asympmorequiv}).
\end{Definition}

Note that, thanks to the fact that $\bp^{-1}\bpsi$ is an automorphism of $\cA$, this is actually an
equivalence relation.

\begin{Proposition}\label{prop:pairasympequiv}
Assume that $\cA$ has a convergent scaling limit and that $\cA$ and $\Aoi$ are isomorphic as C$^*$-algebras, through a properly supported isomorphism.

If $\big((\r_\l),(\phi_\l)\big), \big((\s_\l),(\psi_\l)\big) \in \AMor_\iota(\cA) \times \AIso_\iota(\cA)$ are asymptotically equivalent, then $\r_0 = \Psi\big((\r_\l),(\phi_\l)\big)$ and
$\s_0=\Psi \big((\s_\l),(\psi_\l)\big)$ are unitarily equivalent.

Conversely, if $\r_0, \s_0 \in \Morps(\Aoi)$ are unitarily equivalent, there exist asymptotically
equivalent $\big((\r_\l),(\phi_\l)\big), \big((\s_\l),(\psi_\l)\big) \in \AMor_\iota(\cA) \times \AIso_\iota(\cA)$ such that
$$\r_0 = \Psi\big((\r_\l),(\phi_\l)\big), \quad\s_0 = \Psi \big((\s_\l),(\psi_\l)\big).$$
\end{Proposition}

\begin{proof}
Let $\big((\r_\l),(\phi_\l)\big), \big((\s_\l),(\psi_\l)\big) \in \AMor_\iota(\cA) \times \AIso_\iota(\cA)$ be asymptotically equivalent. Then according to the remarks following Definition~\ref{def:asympmorequiv},
there exists a unitary operator $U_0 \in \Aoi$ such that
\begin{equation*}
\br \bp^{-1}\bpsi= U_0\bs(\cdot)U_0^*, 
\end{equation*}
where, as usual, $\br = \poiac\rac$, $\bs = \poiac\bs^\bullet$, $\bp = \poiac\pac$, $\bpsi = \poiac\bpsi^\bullet$.
Therefore, according to formula~\eqref{eq:difference}, one has, for all $A \in \Aoi$,
\begin{equation*}
\r_0(A) = U_0\bs \bpsi^{-1}(A)U_0^* = U_0 \s_0(A)U_0^*
\end{equation*}
which shows the first part of the statement.

Assume now that $\r_0, \s_0 \in \Morps(\Aoi)$ are unitarily equivalent, and consider a properly supported isomorphism
$\bp : \cA \to \Aoi$. Then, thanks to Theorem~\ref{thm:difference}, there exist $(\r_\l), (\s_\l) \in \AMor(\cA)$
and $(\phi_\l) \in \AIso(\cA)$ such that $\bp = \poiac\pac$ and
$\r_0 = \Psi\big((\r_\l),(\phi_\l)\big)$, $\s_0 = \Psi\big((\s_\l),(\phi_\l)\big)$.
It then follows from formula~\eqref{eq:difference} that $\br = \poiac\rac$, $\bs = \poiac\bs^\bullet$ are
unitarily equivalent, and therefore, again using the remark following Definition~\ref{def:asympmorequiv},
$(\r_\l) = (\r_\l\bp^{-1}\bp)$ and $(\s_\l)$ are asymptotically equivalent.
\end{proof}

Furthermore, given $\big((\r_\l),(\phi_\l)\big), \big((\s_\l),(\psi_\l)\big) \in \AMor_\iota(\cA) \times \AIso_\iota(\cA)$, one sees at once, again by the remarks following Definition~\ref{def:asympmorequiv}, that $\Psi\big((\r_\l),(\phi_\l)\big)= \Psi\big((\s_\l),(\psi_\l)\big)$ if and only if the asymptotic morphisms $(\r_\l\bp^{-1}\bpsi)$ and $(\s_\l)$ are asymptotically
equal.

We summarize the above discussion in the following result.

\begin{Corollary}
If $\cA$ has a convergent scaling limit, there is a bijective correspondence between the asymptotic equivalence classes of $\AMor_\iota(\cA) \times
\AIso_\iota(\cA)$ and the unitary equivalence classes of $\Morps(\Aoi)$.
\end{Corollary}

\begin{proof}
The map which associates to the class of $\big((\r_\l),(\s_\l)\big) \in \AMor_\iota(\cA) \times
\AIso_\iota(\cA)$ the class of $\Psi\big((\r_\l),(\s_\l)\big) \in \Morps(\Aoi)$ is well defined thanks
to the first part of Proposition~\ref{prop:pairasympequiv}, injective thanks to the second part
of the same proposition together with the remark following it, and surjective thanks to the
first part of Theorem~\ref{thm:difference}.
\end{proof}

\section{Localized asymptotic morphisms}\label{sec:localized}
Throughout this section, we assume that $\cA$ has a convergent scaling limit. We now wish to focus our attention on pairs of asymptotic morphisms $\big((\r_\l),(\phi_\l)\big)$ giving rise, via the map $\Psi$, to \emph{localized} morphisms of $\Aoi$. We recall that a morphism $\r$ of a local net $\cA$ is localized if there exists a double cone $O$ such that $\r(A) = A$ if $A \in \cA(O')$. A localized morphism is properly supported if $\cA$ satisfies Haag duality.
A natural guess would be that $\Psi\big((\r_\l),(\phi_\l)\big)$ being localized is equivalent to requiring
\begin{equation*}
\lim_\k \| [\r_{\l_\k}(A) - \phi_{\l_\k}(A)]\O \| = 0, \qquad A \in \cA,\, \bp(A) \in \Aoi(O').
\end{equation*} 
There are however a number of technical problems in proving this, including poor control on $\bp^{-1}(\Aoi(O'))$.

Instead of tackling this issue directly, here we appeal to the characterization of localized morphisms of $\Aoi$ in terms
of asymptotic charge transfer chains (ACTC) of $\cA$~\cite{M1, M2}. These are suitable sequences $(\uU_k)_{k \in \bN} \subset \uAA$  which behave, roughly speaking, as charge transfer chains (CTC)\footnote{The notion of charge transfer chain can be traced back to~\cite{DHR} (see also~\cite{R}). Here we follow the formal definitions and results of~\cite{M1,M2}.} for $\l \to 0$, i.e., $\uU_{k,\l}$ is asymptotically 
bilocalized in a couple of regions, one of which goes to spacelike infinity as $k \to +\infty$, while
$\uU_{k,\l}^*\uU_{h,\l}$ is asymptotically bilocalized 
in a couple of regions both going to spacelike infinity as $k,h \to \infty$. Moreover one has
\begin{equation*}
\r_0^{\uU}(A) = \slim_{k\to\infty}\poi(\uU_k)A\poi(\uU_k)^*, \qquad A \in \Aoi,
\end{equation*}
with $\r_0^\uU$ a localized morphism of $\Aoi$, and all such morphisms arise in this way. We will call $\r_0^\uU$ the scaling limit morphism associated to $(\uU_k)$, and we will denote by $\mathrm{ACTC}(\cA)$ the set of asymptotic charge transfer chains of $\cA$.

Since, as shown in the previous section, there is a one-to-one correspondence
between classes of morphisms of $\Aoi$, localized or not, and classes of pairs of asymptotic morphisms of $\cA$, this will provide us with a description of those pairs of asymptotic morphisms corresponding to localized morphisms in terms of asymptotic charge transfer chains.

\begin{Definition}
We say that a pair $\big((\r_\l),(\phi_\l)\big) \in \AMor_\iota(\cA) \times\AIso_\iota(\cA)$
 is \emph{localized} if $\Psi\big((\r_\l),(\phi_\l)\big) \in \Morps(\Aoi)$ is a localized morphism of $\Aoi$.
\end{Definition}

\begin{Proposition}
The pair $\big((\r_\l),(\phi_\l)\big) \in \AMor_\iota(\cA) \times\AIso_\iota(\cA)$ is localized if and only if
there exists an ACTC $(\uU_k)_{k \in \bN}$ such that
\begin{equation}\label{eq:localasympmor}
\slim_{k \to +\infty} \poiac(\uU_k \pac(A)\uU_k^*) = \poiac\rac(A), \qquad A \in \cA.
\end{equation}
\end{Proposition}

\begin{proof}
Assume that $\r_0 := \Psi\big((\r_\l),(\phi_\l)\big) \in \Morps(\Aoi)$ is localized. Then there exists an ACTC
 $(\uU_k)_{k \in \bN}$ such that
 \begin{equation*}
\slim_{k \to +\infty} \poi(\uU_k) A_0 \poi(\uU_k)^* =  \r_0(A_0), \qquad A_0 \in \Aoi.
\end{equation*}
Thanks to Theorems~\ref{thm:asympmortomor}, \ref{thm:difference}, this implies immediately
\begin{equation*}\begin{split}
\slim_{k \to +\infty} \poiac(\uU_k \pac(A)\uU_k^*) &=  \slim_{k \to +\infty} \poi(\uU_k)\poiac(\pac(A)) \poi(\uU_k)^*\\
&= \r_0(\poiac(\pac(A)) ) = \r_0\bp(A) = \poiac\rac(A).
\end{split}\end{equation*}

Conversely, if formula~\eqref{eq:localasympmor} holds for some ACTC $(\uU_k)_{k \in \bN}$, then
the same computation shows that $\poiac\rac = \r_0^{\uU}\bp$, i.e.\ $\Psi\big((\r_\l),(\phi_\l)\big) = \r_0^{\uU}$ is a localized morphism of $\Aoi$.
\end{proof}

Obviously, there is a map from classes of asymptotic charge transfer chains of $\cA$ to
classes of pairs of asymptotic morphisms of $\cA$ which makes the following diagram commutative:
\begin{equation*}
\begin{array}{ccc}
\D(\Aoi)/\cong     & \hookrightarrow & \Morps(\Aoi)/\cong \\
\updownarrow   &               & \updownarrow \\
\mathrm{ACTC}(\cA)/\sim & \to &\AMor_\iota(\cA)\times\AIso_\iota(\cA)/\sim 
\end{array}
\end{equation*}
where we have denoted by $\D(\Aoi)$ the semigroup of localized morphisms of the net $\Aoi$. Thanks to the above proposition, we see that a class $[(\uU_k)_{k \in \bN}] \in \mathrm{ACTC}(\cA)/\sim$
gets mapped onto a class $\big[ \big((\s_\l),(\psi_\l)\big)\big] \in \AMor_\iota(\cA)\times\AIso_\iota(\cA)/\sim$,
if and only if there exists a unitary operator $W$ on $\Hoi$ such that
\begin{equation*}
\slim_{k \to +\infty} \poiac(\uU_k \bpsi^\bullet(A)\uU_k^*) = W\poiac\bs(A)W^*, \qquad A \in \cA.
\end{equation*}

By the above, there is also a natural notion of transportability of localized pairs of asymptotic morphisms. We skip the easy details.

Of course, in view of applications to superselection theory, an important question that we plan to address elsewhere is to find some characterization of further properties of ``physical morphisms'' (like being irreducible, covariant, with finite statistics, having a conjugate...) of $\Aoi$ in terms of suitable pairs of asymptotic morphisms.

\section{Asymptotic morphisms and preserved sectors}\label{sec:preserved}
We now want to discuss some relation between the asymptotic morphisms of $\cA$ introduced in Theorem~\ref{thm:asympmor} and the morphism $\rho_0$ of $\Aoi$ associated to a preserved sector as defined in~\cite{DMV}. These are the sectors of $\cA$ which ``survive'' to the scaling limit, and thus give rise to sectors of $\Aoi$. Throughout this section, we assume that $d=3,4$.

To this end, we need to introduce the canonical DR field net $\cF$ determined by the superselection sectors of $\cA$~\cite{DR}. The corresponding scaling algebra $\uFF$ and scaling limit net $\Foi$ have been defined in~\cite{DMV}. There, also the concept of (bounded) functions $\l \mapsto F_\l \in \cF$, $F_\l \in \cF(\l O)$, asymptotically contained in $\Foi(\hat O)$ is defined, in a way analogous to Definition~\ref{def:asymcont}.
Using then arguments similar to those of Section~\ref{sec:liftrep}, it is not difficult to show that
the set $\uFFac(O)$ of elements asymptotically contained in $\Foi(\hat O)$ for all $\hat O \supset \bar O$ is a C$^*$-algebra, and that the scaling limit representation $\poi : \uFF \to B(\cH_{\Foi})$ lifts to a representation $\poiac : \uFFac \to B(\cH_{\Foi})$ such that, for all $F \in \uFFac(O)$, and for all $\d$-sequences $(h_\n)_{\n \in \bN}$,
\begin{equation*}
\sslim_{\n \to +\infty}\poi(\ua_{h_\n}F) = \poiac(F).
\end{equation*}
Furthermore the restriction of $\poiac$ to $\uAAac$ and $\Hoi$ coincides with the extension to $\uAAac$ of the restriction of $\poi$ to $\uAA$ and $\Hoi$, cf.~\cite[p.\ 491]{CM}. We will still use for these restrictions the notations $\poiac$, $\poi$ respectively.
Denoting by $g \in G \mapsto \b_g \in \aut(\cF)$ the action of the canonical gauge group $G$ on $\cF$, it is also easy to verify that
\begin{equation*}
\uFFac(O)^G := \{ F \in \uFFac(O)\,:\,\b_g(F_\l) = F_\l\,\,\forall g \in G,\l>0\} = \uAAac(O).
\end{equation*}

We first mention a result similar in spirit to~\cite[prop.\ 5.6]{DMV} but based on
the notion of ACTC as discussed in~\cite{M1, M2}. This is an easy consequence
of the above definitions and known facts. We put it on record here since its formulation,
in which only observables appear, can be subject to interesting generalizations.

\begin{Proposition}\label{prop:ACTCpreserved}
Let $\xi$ be a preserved sector of the local net $\cA$. For any double cone $O$, there exists an ACTC $(\uU_k)_{k \in \bN}\subset \uAA$
such that:
\begin{proplist}{3}
\item \label{it:limukl}there holds, for $A \in \cA$ and $\l > 0$,
\begin{equation}\label{eq:limukl}
\slim_{k \to +\infty} \uU_{k,\l}A\uU_{k,\l}^* = \rho(\l) (A), 
\end{equation}
where $\rho(\l)$ are morphisms of $\cA$ of class $\xi$ localized in $\l O$;
\item \label{it:limpouk}there holds, for $A_0 \in \Aoi$,
\begin{equation*}
\slim_{k\to +\infty} \poi(\uU_k) A_0 \poi(\uU_k)^ *= \rho_0(A_0)
\end{equation*}
with $\rho_0 = \rho_0^\uU$ a morphism of $\Aoi$ localized in any $\hat O \supset \bar O$;
\item \label{it:limahur} for each $A \in \uAAac$ the function $\r^\bullet(A) : \l \mapsto \r(\l)(A_\l)$ belongs to $\uAAac$ and
\begin{equation}\label{eq:poiacrhoAbul}
\poiac(\r^\bullet(A))=\r_0(\poiac(A)).
\end{equation}

\end{proplist}
\end{Proposition}

\begin{proof}
According to~\cite[def.\ 5.4]{DMV}, for each $O$ we can find a multiplet $\psi_j(\l) \in \cF(\l O)$, $j=1,\dots,d$, of class $\xi$ such that the functions $\psi_j:\l \mapsto \psi_j(\l)$ belong to $\uFFac(O)$. Let now $x_k \in W$ be a sequence converging to spacelike infinity in a given wedge $W$ and $(h_k)_{k \in \bN}$ be a $\delta$-sequence of non-negative continuous, compactly supported functions on $\bR^d$. Define
\begin{equation}\label{eq:ACTCpreserved}
\uU_{k,\l} = \sum_{j=1}^d (\ua_{h_k}\psi_j)_\l \ua_{x_k}(\ua_{h_k}\psi_j)^*_\l.
\end{equation}
Then properties \ref{it:limukl} and \ref{it:limpouk} follow from~\cite[thm.\ 4.4]{M2}. This also implies that
indeed
\begin{equation*}
\r(\l)(A) := \sum_{j=1}^d\psi_j(\l)A\psi_j(\l)^*, \qquad A \in \cA,
\end{equation*}
and, using also~\cite[prop.\ 5.5]{DMV} and Lemma~\ref{lem:poiextendlim},
\begin{equation}\label{eq:preservedmor}
\r_0(A_0) = \sum_{j=1}^d\poiac(\psi_j)A_0\poiac(\psi_j)^*, \qquad A_0 \in \Aoi.
\end{equation}
The localization property of $\r_0$ is then a consequence of the fact that $\poiac(\uFFac(O)) \subset \cF_{0,\iota,r}(O)$, cf.\ the proof of Thm.~\ref{thm:poiextend}.
For what concerns property~\ref{it:limahur}, note that since $A \in \uAAac \subset \uFFac$, $\psi_j \in \uFFac$, $j=1,\dots, d$, and $\uFFac$ is a C$^*$-algebra, it follows at once that $\r^\bullet(A) \in \uFFac$. It is then also clear that $\b_g(\r(\l)(A_\l)) = \r(\l)(A_\l)$ for all $g\in G$, $\l > 0$, and therefore by the above remark $\r^\bullet(A) \in \uAAac$. Then equation~\eqref{eq:poiacrhoAbul} follows immediately from the fact that $\poiac$ is a representation of $\uFFac$.
\end{proof}

We now show that, for preserved sectors, there is a preferred choice of the associated asymptotic
morphism.

\begin{Theorem}\label{thm:preserved}
Let $\xi$ be a preserved sector of $\cA$, and $\r_0 \in \Morps(\Aoi)$, $\r(\l) \in \Morps(\cA)$ the associated morphisms as in Prop.~\ref{prop:ACTCpreserved}. If $(\phi_\l) \in \AIso_\iota(\cA)$, defining
\begin{equation}\label{eq:preservedasympmor}
\r_\l := \r(\l)\phi_\l, \qquad \l > 0,
\end{equation}
there holds
that $(\r_\l) \in \AMor_\iota(\cA)$ and $\Psi\big((\r_\l),(\phi_\l)\big) = \r_0$.
\end{Theorem}

\begin{proof}
We start by showing that $(\r_\l)$ is an asymptotic morphism of $\cA$. Since $\r(\l)$ is a morphism
of $\cA$ for all $\l > 0$ one has, for all $A, B \in \cA$,
\begin{equation*}\begin{split}
\lim_\k \big\| \big[\r_{\l_\k}(A)&\r_{\l_\k}(B) - \r_{\l_\k}(AB)\big]\O\| \\
&= \lim_\k\big\| \big[\r(\l_\k)\big(\phi_{\l_\k}(A)\phi_{\l_\k}(B) - \phi_{\l_\k}(AB)\big)\big]\O\| \\
&= \big\| \poiac\big(\r^\bullet(\pac(A)\pac(B)-\pac(AB))\big)\Ooi\big\| \\
&= \big\| \sum_j \poiac(\psi_j)\big(\bp(A)\bp(B)-\bp(AB)\big)\poiac(\psi_j)^* \Ooi \| = 0
\end{split}\end{equation*}
where in the third equality we used Equations~\eqref{eq:poiacrhoAbul}, \eqref{eq:preservedmor} and where the last equality holds because $\bp = \poiac\pac : \cA \to \Aoi$ is an isomorphism. Thus Equation~\eqref{eq:asympmult} holds for
$(\r_\l)$. Equations~\eqref{eq:asympadjoint}, \eqref{eq:asymplinear} are verified similarly.

Next, we prove that $(\r_\l)$ is tame. Property~(i) of Definition~\ref{def:tameasympmor} follows at once from Proposition~\ref{prop:ACTCpreserved}\ref{it:limahur} and from the fact that, since $(\phi_\l)$ is tame, $\pac(A) \in \uAAac$ for all $A \in \cA$. For property (ii), $\r(\l)$ being a morphism of C$^*$-algebras, and therefore contractive, there holds
\begin{equation*}
\sup_\l \| \r(\l)(\phi_\l(A) - \phi_\l(B))\| \leq  \sup_\l \| \phi_\l(A) - \phi_\l(B) \|,
\end{equation*}
and we can then use again that (ii) holds for $(\phi_\l)$. Finally property~(iii) follows from Equations~\eqref{eq:poiacrhoAbul}, \eqref{eq:preservedmor}, from the  localization of $\poiac(\psi_j)$ and from property (iii) for $(\phi_\l)$.

It remains to prove that $\Psi\big((\r_\l),(\phi_\l)\big) = \r_0$. Since, according to Proposition~\ref{prop:ACTCpreserved}\ref{it:limahur}, $\poiac(\r^\bullet(\pac(A)))= \r_0(\bp(A))$, $A \in \cA$, in order to conclude it is sufficient, thanks to formula~\eqref{eq:difference}, to observe that
\begin{equation*}
\poiac(\rac(A)) = \poiac(\r^\bullet(\pac(A))), \qquad A \in \cA,
\end{equation*}
which is clear by definition.
\end{proof}

\begin{Remark}\label{rem:preserved}
\begin{remlist}
 \item According to Definition~\ref{def:pairasympmorequiv} and the remark following Proposition~\ref{prop:pairasympequiv}, if $(\s_\l)$ is
a tame asymptotic morphism such that $\Psi\big((\s_\l),(\phi_\l)\big) = \r_0$ then $(\s_\l)$ is asymptotically equal to $(\r_\l)$ defined by~\eqref{eq:preservedasympmor}.
\item Let $\xi, \r_0$ and $(\phi_\l)$ be as above, and let $(\uU_k)$ be the ACTC defined by equation~\eqref{eq:ACTCpreserved}. Consider, for each $A \in\cA$ and $k \in \bN$, the function $\r_k(A) : \l \mapsto \uU_{k,\l}\phi_\l(A)\uU_{k,\l}^*$, which is an element of $\uAAac$. Then, thanks to Proposition
~\ref{prop:ACTCpreserved}, there holds
\begin{equation}\label{eq:interchange}
\poiac\Big(\lim_{k \to +\infty}\r_k(A)\Big) = \slim_{k \to +\infty}\poiac(\r_k(A)),
\end{equation}
where the limit in the left hand side is to be understood as the pointwise convergence, in the strong operator topology, of
the sequence of functions $(\r_k)_{k \in \bN}$ with values in $\cB(\cH)$. 
As the definition of $\poiac$ involves a $\l \to 0$ limit, Equation~\eqref{eq:interchange} can be interpreted as an exchange of the limits in $k$ and $\l$.
\item \label{it:dilation}If $\cA$ is a dilation covariant theory satisfying the Haag-Swieca compactness condition, it is not difficult to see that the isomorphism $\bp : \cA \to \Aoi$ of~\cite[thm. 7.1]{BDM2} satisfies
\begin{equation*}
\bp^{-1}(\poiac(A)) = \wlim_\k \d^{-1}_{\l_\k}(A_{\l_\k}), \qquad A \in \uAAac,
\end{equation*}
(see also the proof of Proposition~\ref{prop:ACTCdilation} below) and that a continuous section $s : \Aoi \to \uAAac$ is obtained by
\begin{equation*}
s(\poiac(A))_\l := \d_\l\bp^{-1}(A_{\l_\k}), \qquad A \in \uAAac.
\end{equation*}
Then, the asymptotic isomorphism corresponding to $\bp^{-1}$ and to this choice of the section is given by $\phi_\l = \d_\l$, and therefore equation~\eqref{eq:preservedasympmor} gives back equation~\eqref{eq:rholdilation}.
\end{remlist}
\end{Remark}

\section{Some remarks on asymptotic charge transfer chains}\label{sec:ACTCremarks}
It seems interesting to understand what becomes of Proposition~\ref{prop:ACTCpreserved} and Theorem~\ref{thm:preserved} of
the previous section if one drops the assumption that $\xi$ is a preserved sector of $\Aoi$. In particular one may ask under which conditions on an ACTC $(\uU_k)_{k \in \bN} \subset \uAA$  the limit in the left hand side of~\eqref{eq:limukl} exists, what kind of  maps are obtained in this way, and what is their relation with the morphism $\r_0^\uU$ associated to $(\uU_k)$. We point out here some facts that illustrate what kind of
difficulties may arise in trying to answer such questions.

\begin{Definition}
The map $\Phi_\mu : \uAA \to \uAA$, $\mu > 0$, defined by
\begin{equation*}
\Phi_\mu(\uA)_\l := \begin{cases}\uA_\l \quad&\l < \mu,\\
                                                           1 &\l \geq \mu,\end{cases}
\end{equation*}
is called a filter at scale $\mu$.
\end{Definition} 

If $\uooi$ is a scaling limit state, with associated scaling limit representation $\poi$,
there holds $\poi(\Phi_\mu(\uA)) = \poi(\uA)$ for all $\uA \in \uAA$, $\mu > 0$.

\begin{Proposition}
Let $(\uU_k)_{k \in \bN}$ be an ACTC of $\cA$.
Then there exists an ACTC $(\uV_k)_{k \in \bN}$ with $\rho_0^\uV=\rho_0^\uU$ and such that, for all $\l > 0,$
\begin{equation*}
\lim_{k \to +\infty} \uV_{k,\l} A \uV_{k,\l}^* = A, \qquad A \in \cA.
\end{equation*}
\end{Proposition}

\begin{proof}
Define $\uV_k := \Phi_{1/k}(\uU_k) \in \uAA$. Then by the above remark $\poi(\uV_k) = \poi(\uU_k)$
for all $k$. According to~\cite[def.\ 4.1]{M2} this implies that $(\uV_k)_{k \in \bN}$ is an ACTC, and it is also clear that it has
$\rho_0^\uU$ as associated scaling limit morphism. In addition, for $k > \l^{-1}$ one
has
\begin{equation*}
\uV_{k,\l} A \uV_{k,\l}^* = A,
\end{equation*}
which concludes the proof.
\end{proof}

An ACTC $(\uV_k)_{k \in \bN}$ as above will be called trivial. Therefore every localized transportable
morphism of the scaling limit theory is associated to a trivial ACTC, irrespective of whether it belongs to a confined sector or not.
Furthermore, since $(\uV_k)_k$ and $(\uU_k)_k$ give rise to the same scaling limit
localized morphism, they are equivalent ACTCs in the sense of~\cite{M1, M2}.

\begin{Remark}
\begin{remlist}
\item Going back to the case of preserved sectors considered in Proposition~\ref{prop:ACTCpreserved}, we see that it is possible to find pairs of ACTCs $(\uU_k)_k$,
$(\uV_k)_k$ which have the same associated morphism in the scaling limit and which define morphisms $\rho(\l)$, $\sigma(\l)$ at finite scales which are non-equivalent for every $\l > 0$.
\item In general, a trivial ACTC does not satisfy condition~\eqref{eq:poiacrhoAbul} in the previous section since, in this case, $\r^\bullet(A) = A$ for all $A \in \uAAac$ 
\end{remlist}
\end{Remark}

Elaborating on the above argument a bit more we can show that, at least for dilation covariant theories, one can change an ACTC, within its equivalence class, in such a way as to obtain an arbitrary sector in the $k \to \infty$ limit (at fixed scale). However, this unpleasant feature is ruled out by requiring that the morphisms at fixed scales and in the scaling limit are connected by Equation~\eqref{eq:poiacrhoAbul}.

\begin{Proposition}\label{prop:ACTCdilation}
Let $\cA$ be a dilation covariant net satisfying Haag-Swieca compactness. Given  an ACTC $(\uV_k)_k$ and a localized transportable morphism $\r$ belonging to a dilation covariant sector $\xi$ of $\cA$,
there exists an ACTC $(\uW_k)_{k \in \bN}$ with $\rho_0^\uW = \rho_0^\uV$  and such that, for all $\l > 0,$
\begin{equation}\label{eq:limWkl}
\slim_{k \to +\infty} \uW_{k,\l} A \uW_{k,\l}^* = \d_\l \r \d_\l^{-1}(A), \qquad A \in \cA.
\end{equation}
Furthermore if $\bp^{-1} : \Aoi \to \cA$ is the isomorphism defined in~\cite{BDM2}, one has $\rho^\uV_0 = \bp\r \bp^{-1}$ if and only if condition~\eqref{eq:poiacrhoAbul} holds with $\rho_0 = \rho_0^\uV$, $\r(\l) = \d_\l\r\d_\l^{-1}$, $\l > 0$ .
\end{Proposition}

\begin{proof}
Since the sector $\xi$ is preserved \cite[Sec. 5]{DMV}, we can apply Prop.~\ref{prop:ACTCpreserved}. In particular, we get an ACTC $(\uU_k)_{k \in \bN}$ such that \eqref{eq:limukl} holds with $\r(\l) = \d_\l \r \d_\l^{-1}$.
Define then
\begin{equation*}
\uW_{k,\l} := \begin{cases} \uV_{k,\l} \quad &\l < 1/k,\\
                                               \uU_{k,\l} &\l \geq 1/k.
                      \end{cases} 
\end{equation*} 
It is clear that $\uW_k \in \uAA$ and that $\poi(\uW_k) = \poi(\uV_k)$, so that $(\uW_k)_k$ is an ACTC and $\rho_0^\uU$ is associated also to $(\uW_k)_k$. Eq.~\eqref{eq:limWkl} is also obvious.

Moreover, we recall from~\cite{BV} that $\bp^{-1}$ is defined first by $\bp^{-1}(\poi(\uA)) = \wlim_\k \d_{\l_\k}^{-1}(\uA_{\l_\k})$, $\uA \in \uAA$. Then by an $\varepsilon/3$ argument, together with Lemma~\ref{lem:poiextendlim} and the fact that $\bp^{-1}$ is unitarily implemented, we see that there also holds
\begin{equation*}
\bp^{-1}(\poiac(A)) = \wlim_\k \d_{\l_\k}^{-1}(A_{\l_\k})
\end{equation*}
for all $A \in \uAAac$. Then since $\r^\bullet(A)_\l = \d_\l\r\d_\l^{-1}(A_\l)$, one has, for $A \in \uAAac(O)$,
\begin{equation*}\begin{split}
\bp^{-1}(\poiac(\r^\bullet(A))) &= \wlim_\k \d_{\l_\k}^{-1}(\r^\bullet(A)_{\l_\k}) = \wlim_\k \r\d_{\l_\k}^{-1}(A_{\l_\k})\\
&= \r\bp^{-1}(\poiac(A)),
\end{split}\end{equation*}
where in the last equality we used the local normality of the localized transportable morphism $\r$ and the fact that $\d_\l^{-1}(A_\l) \in \cA(O)$ for all $\l > 0$. Therefore equation~\eqref{eq:poiacrhoAbul} holds if and only if $\rho^\uV_0 = \bp\r \bp^{-1}$.
\end{proof}

More generally, consider a family $(\sigma(\l))_{\l>0}$ of morphisms of a local net $\cA$ and a sequence $(\uV_k)_k \subset \uAA$ such that, in some of the usual topologies,
\begin{equation*}
\lim_{k \to +\infty} \uV_{k,\l} A \uV_{k,\l}^* = \sigma(\l)(A), \qquad A \in \cA.
\end{equation*}
It is then clear, again by the same argument, that for any scaling limit morphism $\rho_0$ one can find an ACTC $(\uW_k)_k$ which has $\rho_0$ as associated scaling limit morphism and such that
\begin{equation}\label{eq:limW}
\lim_{k \to +\infty} \uW_{k,\l} A \uW_{k,\l}^* = \sigma(\l)(A), \qquad A \in \cA,
\end{equation}
in the same topology as the above limit.
\medskip

Due to the large arbitrariness in the values
of $\uU_k$ at finite scales it seems difficult, if not impossible, to show, for a general ACTC $(\uU_k)_k$, that there exist the limit
\begin{equation*}
\lim_{k\to +\infty}\uU_{k,\l} A \uU_{k,\l}^*
\end{equation*}
for each fixed $\l > 0$. This arbitrariness is also at the root of the results presented in this section. They indicate that, even if the above limit exists, there is in general no tight connection between the maps it defines at each scale $\l > 0$ and the scaling limit morphism $\r_0^\uU$ associated to $(\uU_k)$. As in the case of preserved sectors, in order to have such a connection one has to select a subclass of ACTCs possessing suitably strengthened properties.

\section{A C$^*$-category of asymptotic morphisms}\label{sec:category}
We now wish to
define a C$^*$-category whose objects are pairs of asymptotic morphisms
of $\cA$ in $\AMor_\iota(\cA)\times\AIso_\iota(\cA)$. In order to do this, we proceed
in parallel with appendix~\ref{app:category}, where we perform the same construction using Connes-Higson
asymptotic morphisms as objects. As many details are common to both situations, we
limit ourselves here to indicating the main steps.

\begin{Definition}
An intertwiner between two pairs of asymptotic morphisms $\big((\r_\l),(\phi_\l)\big)$ and $\big((\s_\l),(\psi_\l)\big)$ in
$\AMor_\iota(\cA)\times \AIso_\iota(\cA)$ is a function $T  \in \uAAac$ such that, for all $A \in \Aoi$,
\begin{equation*}
\lim_\k \| [T_{\l_\k}\r_{\l_\k}\bp^{-1}(A) - \s_{\l_\k}\bpsi^{-1}(A)T_{\l_\k}]\O \| = 0.
\end{equation*}
We write $T \in \hom\big[\big((\r_\l),(\phi_\l)\big), \big((\s_\l),(\psi_\l)\big)\big]$.
\end{Definition}

From the above definition, it follows that if $\r_0 = \Psi\big((\r_\l),(\phi_\l)\big)$, $\s_0 = \Psi\big((\s_\l),(\psi_\l)\big)$, then $T \in \hom\big[\big((\r_\l),(\phi_\l)\big), \big((\s_\l),(\psi_\l)\big)\big]$ if and only if $\poiac(T)$ is an intertwiner between $\r_0$ and $\s_0$ in $\Morps(\Aoi)$.

Given intertwiners $T \in \hom\big[\big((\r_\l),(\phi_\l)\big), \big((\s_\l),(\psi_\l)\big)\big]$, $S \in \hom\big[\big((\s_\l),(\psi_\l)\big), \big((\t_\l),(\chi_\l)\big)\big]$, their composition is naturally defined
as $(S\circ T)_\l := S_\l T_\l$ and, thanks to the observation above, it is clear that
$S \circ T \in \hom\big[\big((\r_\l),(\phi_\l)\big), \big((\t_\l),(\chi_\l)\big)\big]$. The linear space of intertwiners $\hom\big[\big((\r_\l),(\phi_\l)\big), \big((\s_\l),(\psi_\l)\big)\big]$, naturally equipped with the norm inherited from $\uAAac$
\begin{equation*}
\| T \| = \sup_{\l>0}\|T_\l \|,
\end{equation*}
becomes a Banach space.
Furthermore $\| S\circ T\| \leq \| S\| \|T\|$. Finally, an involution on
intertwiners is defined by $T^* := (\l \mapsto T_\l^*) \in \hom\big[ \big((\s_\l),(\psi_\l)\big),\big((\r_\l),(\phi_\l)\big)\big]$,
and one straightforwardly verifies that the axioms of a C$^*$-category~\cite{GLR} are satisfied. In particular each self-intertwiner space $\hom\big[\big((\r_\l),(\phi_\l)\big),\big((\r_\l),(\phi_\l)\big)\big]$ becomes a unital C$^*$-algebra.

We denote the C$^*$-category thus obtained by $\cE_\iota(\cA)$.

\begin{Proposition}There is a canonical $*$-functor $\Psi : \cE_\iota(\cA) \to \Morps(\Aoi)$ defined
by the map $\Psi : \AMor_\iota(\cA) \times \AIso_\iota(\cA) \to \Morps(\Aoi)$ of Equation~\eqref{eq:Psi} on objects and by
$\Psi(T) := \poiac(T)$ on intertwiners. 

If $\cA$ has a convergent scaling limit, and if there exists a properly supported isomorphism $\bp : \cA \to \Aoi$, then $\Psi$ is full and surjective.
\end{Proposition}

 Surjectivity of the functor $\Psi$ follows from Theorem~\ref{thm:difference}, and its
 fullness from Theorem~\ref{thm:poiextend}.

 The above statement is of course valid also if we restrict to the full subcategories of
 $\cE_\iota(\cA)$ and of $\Morps(\Aoi)$ defined by the localized objects.
 \medskip

 It is also not difficult to see that it is possible to generalize the above result to asymptotic
 morphisms between two nets $\cA$, $\cB$ and their intertwiners, defined in the obvious way. In this way
 one gets a C$^*$-category $\cE_\iota(\cA,\cB)$ and a full surjective $*$-functor between $\cE_\iota(\cA,\cB)$ and $\Morps(\Aoi,\cB_{0,\iota})$.

\section{Summary and outlook}\label{sec:summary}
In the present investigation, we introduced a variant of the concept of asymptotic morphism of Connes and Higson, and showed that such objects can be used to describe the superselection (i.e., charge) structure
of the short distance scaling limit of a QFT defined by a given local net $\cA$. 

The main differences between our notion of asymptotic morphism and the Connes-Higson one are due to the facts that: (1) the algebra-valued functions which naturally arise in the description of the scaling limit enjoy rather poor continuity properties, and in particular they cannot be norm continuous for $\l \to 0$, and (2) when studying superselection theory, one has to assume that the local algebras are weakly closed, and therefore they cannot be just quotients of the local scaling algebras with respect to the kernel of the scaling limit representation.  

In order to tackle issue (2) above, we had to construct, following ideas already present in~\cite{DMV}, a new net $O \mapsto \uAAac(O)$ which extends the scaling algebras net, and which is still represented on the original scaling limit Hilbert space in such a way as to contain the local von Neumann algebras $\Aoi(O)$ (if the underlying theory has a convergent scaling limit). This allowed us to set up a bijective correspondence between classes of asymptotically equivalent asymptotic morphisms $(\r_\l)$ and unitary equivalence classe of morphisms $\br$ from $\cA$ to $\Aoi$, so that $\br(A)$, $A \in \cA$, is the ``scaling limit'' of the function $\l \mapsto \r_\l(A)$. In order to relate the morphisms $\br$ with the superselection structure of $\Aoi$, we showed that the quasi-local C*-algebras of two nets $\cA$ and $\cB$ are isomorphic under very general conditions. 
Therefore by mapping the pair of asymptotic morphisms $\big((\r_\l),(\phi_\l)\big)$, the second of which corresponds to such an isomorphism $\bp : \cA \to \Aoi$, to the morphism $\r_0 = \br\bp^{-1} : \Aoi \to \Aoi$, we obtain a bijection between classes of such pairs and unitary equivalence classes of morphisms of $\Aoi$.

We note that the existence of an isomorphism $\bp : \cA \to \cB$ at the level of quasi-local C*-algebras by no means implies the physical equivalence of the corresponding theories. Even if $\bp$ respects the net structure, which generally speaking seems not to be the case, there is no reason why it should also intertwine the actions of the symmetry group and the vacuum states. Therefore, it seems an interesting problem whether one can establish,  in the case $\cB = \Aoi$, the existence of such an isomorphism which is canonical in some sense.

We have also shown that one can view (localized) pairs of asymptotic morphisms as above as objects of a C*-category $\cE_\iota(\cA)$. A natural question is then how much of the information carried by the superselection category of $\Aoi$ can be encoded into $\cE_\iota(\cA)$. In particular one may ask if it is possible to define  additional structures/operations on the family of asymptotic morphisms such as, e.g., a tensor product. In this respect, we observe that the obvious composition of asymptotic morphisms is not in general an asymptotic morphism. This is a problem also in the context of Connes-Higson E-theory, where it is solved by appealing to homotopy classes. It seems possible that similar methods can be employed in the present case too.

It could also be worthwhile to study these concepts in models. Apart from the obvious case of free field theories, it should also be remarked that for the methods developed in this work the dimensionality of spacetime is irrelevant, thus opening the way to the study of two-dimensional models, for which a detailed analysis of nontrivial examples is available. In particular, one may ask questions about the nature of the  asymptotic morphisms corresponding to the confined sectors of the Schwinger model~\cite{Buc1, BV2}. (A possible source of problems, here, could be the arbitrariness introduced by the use of a section of $\poi$ in defining them.) 

The study of such examples could also provide hints about a possible characterization of confined sectors in terms of asymptotic morphisms.
 
More in prospect, as already suggested in~\cite{CDM}, bearing in mind the relationships between E-theory and KK-theory one should investigate the possibility of using the asymptotic morphisms introduce here in order to define noncommutative geometric invariants with a clearcut physical interpretation for the local net $\cA$.

\vspace{\baselineskip}
\noindent\emph{Acknowledgements}. We wish to thank S.~Doplicher, who first pointed out the possible relevance of asymptotic morphisms in the description of superselection structure in the scaling limit of QFT. We are also grateful to M.~Weiner for a useful discussion about Sec.~3 and to R.~Hillier for his comments on a preliminary version of this work.

\appendix

\section{Appendix: The $C^*$-category of asymptotic morphisms}
\label{app:category}

Given a unital $C^*$-algebra $\cA$, it is well-known that
there is a $C^*$-category ${\rm End}(\cA)$
whose objects are the (unital $*$-)endomorphisms from $\cA$ into itself
and whose arrows are their intertwiner operators. Recall that
given two endomorphisms $\rho$ and $\rho'$
one says that an element $T \in \cA$ is an intertwiner between $\rho$ and $\rho'$ if it holds 
$$T \rho(A) = \rho'(A) T, \qquad A \in \cA \, , $$
and denotes by $(\rho,\rho')$ the Banach space of all such intertwiners.
The $C^*$-category ${\rm End}(\cA)$ is equipped with relevant additional structure, for instance it is a monoidal $C^*$-category, where the monoidal (tensor) product of objects is given just by composition, $\rho \otimes \rho' := \rho\rho'$, while the monoidal operation of arrows is defined by
$$T \times S := T \rho(S) = \rho'(S) T, \qquad T \in (\rho,\rho'), S \in (\sigma,\sigma')\, . $$
It is then plain to verify that indeed $T \times S \in (\rho\rho',\sigma\sigma')$.

\medskip

Similarly, if $\cA$ and $\cB$ are two unital $C^*$-algebras one may define the
$C^*$-category ${\rm End}(\cA,\cB)$, whose objects are the endomorphisms from $\cA$ into $\cB$ and whose arrows are the intertwining operators, i.e. elements of $\cB$ satisfying the same intertwining relation as above.
However the tensor structure is lost when $\cB \neq \cA$ as the composition of endomorphisms is no longer possible.

\medskip
Our aim in this Appendix is to extend the previous construction to the case of Connes-Higson asymptotic morphisms. According to E-theory, as developed in~\cite{CH},
an asymptotic morphism between two unital $C^*$-algebras $\cA$ and $\cB$
is a family $(\rho_\lambda)_{\lambda > 0}$ such that each $\rho_\lambda$, $\lambda>0$, is a map from $\cA$ into $\cB$, satisfying the following properties:

\begin{proplist}{2}
\item for each $A \in \cA$,
the function $\lambda \mapsto \rho_\lambda(A)$ belongs to $C_b({\mathbb R}_+,\cB)$, the bounded continuous functions from ${\mathbb R}_+$ to $\cB$;
\item it holds, for all $A, B \in \cA$ and $\a \in {\mathbb C}$,
\begin{align*}
\lim_{\lambda \to 0} & \|\rho_\lambda(A^*) - \rho_\lambda(A)^*\| = 0 \ , \\ 
\lim_{\lambda \to 0} & \|\rho_\lambda(A+\a B) - \rho_\lambda(A) - \a \rho_\lambda(B)\| = 0, \\
\lim_{\lambda \to 0} & \|\rho_\lambda(A B) - \rho_\lambda(A)\rho_\lambda(B)\| = 0 \ .
\end{align*}
\end{proplist}

It is then easy to see that every asymptotic morphism $(\rho_\lambda)$ from $\cA$
into $\cB$ defines a {\it bona fide} morphism $\br$ from $\cA$ into $\cB_0 := C_b({\mathbb R_+},\cB)/C_0({\mathbb R}_+,\cB)$ given by
$$\br (A) = (\lambda \mapsto \rho_\lambda(A)) + C_0({\mathbb R_+},\cB) \, , $$
where $C_0({\mathbb R}_+,\cB)$ is the closed ideal of $C_b({\mathbb R_+},\cB)$
of functions vanishing at the origin.

We now wish to define a $C^*$-category ${\mathcal E}(\cA,\cB)$  of asymptotic morphisms, namely a $C^*$-category whose objects are asymptotic morphisms from $\cA$ into $\cB$ and whose arrows are suitably defined asymptotic intertwiners.

\begin{Definition}
Given two asymptotic morphisms $(\rho_\lambda)$ and $(\rho'_\lambda)$ from $\cA$ into $\cB$, an asymptotic intertwiner is a family $(T_\lambda)$ of elements of $\cB$, indexed by ${\mathbb R_+}$, such that:
\smallskip
\begin{proplist}{2}
\item the function $\lambda \mapsto T_\lambda$ belongs to $C_b({\mathbb R}_+,\cB)$;
\item $\lim_{\lambda\to 0} \|T_\lambda \rho_\lambda(A) - \rho'_\lambda(A)T_\lambda\| = 0$, for all $A \in \cA$. 
\end{proplist}
\end{Definition}
We then write $(T_\lambda) \in \hom\big[(\rho_\lambda),(\rho'_\lambda)\big]$. Notice that two asymptotic
morphisms $(\r_\l),(\r'_\l)$ are asymptotically equivalent in the sense of~\cite{CH} if and
only if $\Id \in \hom\big[(\r_\l),(\r'_\l)\big]$, where $\Id : \bR_+ \to \cB$ denotes the constant function $\Id_\l = 1$.

We now claim that using the above definition one obtains a $C^*$-category whose
objects and arrows are asymptotic morphisms and asymptotic intertwiners, respectively.

For this purpose let us define the norm of the asymptotic intertwiner $(T_\lambda)$ by $\|(T_\lambda)\| = \sup_{\lambda > 0} \|T_\lambda\|$.
Now it is straightforward to verify the following facts:

\begin{proplist}{3}

\item
composition of asymptotic intertwiners is well-defined and associative:
namely, $(T_\lambda)(S_\lambda) \in \hom\big[(\rho_\lambda),(\rho''_\lambda)\big]$
whenever $(S_\lambda) \in \hom\big[(\rho_\lambda),(\rho'_\lambda)\big]$ and $(T_\lambda) \in \hom\big[(\rho'_\lambda),(\rho''_\lambda)\big]$,
where the composition is defined as $(T_\lambda)(S_\lambda) := (T_\lambda S_\lambda)$
and $(U_\lambda) (T_\lambda) (S_\lambda)$ is independent of the bracketing; moreover, for each $(\r_\l)$ there is a unit $\Id \in \hom\big[(\r_\l),(\r_\l)\big]$;

\item each asymptotic intertwiner space $\hom\big[(\rho_\lambda),(\rho'_\lambda)\big]$ is a complex linear space which, when equipped with the above norm, becomes a Banach space, and $\| (T_\l)(S_\l) \| \leq \| (T_\l) \| \| (S_\l)\|$;

\item
on the (Banach) category obtained in this way there is a contravariant antilinear involutive endofunctor $*$,
given by $(T_\lambda)^* = (T_\lambda^*)$,
such that $\hom\big[(\rho_\lambda),(\rho'_\lambda)\big]^* = \hom\big[(\rho'_\lambda),(\rho_\lambda)\big]$;

\item
the $C^*$-identity for the norm holds, namely $\|(T_\lambda)^* (T_\lambda)\| = \|(T_\lambda)\|^2$
for any asymptotic intertwiner $(T_\lambda)$;

\item if $(S_\lambda) \in \hom\big[(\rho_\lambda),(\rho'_\lambda)\big]$ then $(S_\lambda^* S_\lambda)$
is a positive element of the $C^*$-algebra $C_b({\mathbb R}_+,\cB)$, and hence of its (unital) C$^*$-subalgebra $\hom\big[(\rho_\lambda),(\rho_\lambda)\big]$.

\end{proplist}

One can also show that if $(T_\lambda)$ is an asymptotic intertwiner between the
asymptotic morhisms $(\rho_\lambda)$ and $(\sigma_\lambda)$ from $\cA$ into $\cB$ then it defines naturally a {\it bona fide} intertwiner $T_0$ between the corresponding morphisms $\br$ and $\bs$ from $\cA$ into $\cB_0$. In this way one obtains a $*$-functor $F$ from the $C^*$-category ${\mathcal E}(\cA,\cB)$ into the $C^*$-category ${\rm Mor}(\cA,\cB_0)$, which is easily seen to be surjective on the objects and full.

We collect together all the facts discussed above.
\begin{Theorem}
Let $\cA$ and $\cB$ be unital $C^*$-algebras.
\begin{proplist}{2}
\item
With the above definitions, ${\mathcal E}(\cA,\cB)$ is a $C^*$-category.
\item
There is a canonical surjective full $*$-functor $F:{\mathcal E}(\cA,\cB) \to {\rm Mor}(\cA,\cB_0)$.
\end{proplist}
\end{Theorem}

In particular, there are natural surjective $*$-functors ${\mathcal E}(\cA,\cA) \to {\rm Mor}(\cA,\cA_0)$
and ${\mathcal E}(\cA_0,\cA) \to {\rm End}(\cA_0)$.

It seems interesting to study the existence of further remarkable structures of the C$^*$-categories $\cE(\cA,\cB)$, but this falls outside the scope of the present work.

\end{document}